\documentclass[review,3p,times,authoryear]{elsarticle}

\makeatletter
\def\ps@pprintTitle{%
  \let\@oddhead\@empty
  \let\@evenhead\@empty
  \let\@oddfoot\@empty
  \let\@evenfoot\@empty
}
\makeatother

\usepackage{amssymb}
\usepackage{amsmath}
\usepackage{amsthm}
\usepackage{mathtools}
\usepackage{float}
\usepackage{booktabs}
\usepackage{comment}
\usepackage{graphicx}
\usepackage{subcaption}
\usepackage{url}
\usepackage{xcolor}
\usepackage{tikz}

\newtheorem{theorem}{Theorem}
\newtheorem{proposition}{Proposition}

\newtheorem{definition}{Definition}
\usepackage{siunitx}
\usepackage{threeparttable}
\usepackage{caption}
\usepackage{multirow}
\usepackage{tabularx}

\newcommand{\Tau}{\mathsf{T}}

\journal{}

\begin{document}

\begin{frontmatter}

\title{How optimistic inflow forecasts distort dispatch, prices, and contracts in hydro-dominated power systems: evidence from Brazil}

\author[a]{Arthur Brigatto}
\author[a,b]{Alexandre Street}
\author[a,c]{Joaquim Dias Garcia}

\affiliation[a]{organization={LAMPS, Pontifical Catholic University of Rio de Janeiro},
            city={Rio de Janeiro},
            country={Brazil}}

\affiliation[b]{organization={Stanford University},
            city={Stanford, CA},
            country={USA}}

\affiliation[c]{organization={PSR},
            city={Rio de Janeiro},
            country={Brazil}}            

\begin{abstract}
Centralized hydrothermal planning models determine generation schedules and electricity spot prices based on inflow forecasts in audited-cost power systems, such as those prevalent in Latin America, and provide operational benchmarks and decision support in hydro-dominated competitive electricity markets. Consequently, biased forecasts can propagate directly into both operational decisions and market outcomes. This paper studies how persistent optimistic inflow forecast bias can propagate through a hydrothermal power system, from system operation to market outcomes. For a stylized hydrothermal model, we show that optimistic bias weakly reduces water values and weakly increases first-stage hydro discharge relative to the unbiased optimum, thereby lowering reservoir storage and postponing thermal commitment. Using official Brazilian planning and operational data, we provide empirical evidence consistent with this mechanism. We then conduct controlled SDDP experiments to compare policies trained under biased and bias-corrected inflow-forecast processes, evaluating both under the same bias-corrected inflow scenarios. The policy based on biased forecasts produces lower reservoir levels, delayed dry-season thermal dispatch, sharper spot-price peaks, higher reliability risk, and higher expected operating costs. Finally, a stylized contracting exercise indicates that these distortions increase the price–quantity risk of a representative hydropower producer and reduce its willingness-to-contract. The results indicate that inflow-forecast bias can be a source of operational inefficiency, reliability risk, and distorted market incentives in hydro-dominated power systems. The insights and policy implications drawn in this paper may be relevant to hydro-dominated systems and electricity markets that are increasingly reliant on energy storage.

\end{abstract}


 \end{frontmatter}

\section{Introduction} \label{sec:introduction}
According to recent reports \citep{IHA_WorldHydropowerOutlook2025, IEA_Hydroelectricity2025}, hydropower is still one of the world’s largest renewable energy sources. Large hydroelectric systems with significant storage capacity operate in Brazil, Canada, Chile, China, Colombia, New Zealand, Norway, and parts of the United States. Beyond their large share, these plants play a strategic role in system reliability by providing the operational flexibility needed to accommodate the large-scale integration of distributed and centralized intermittent renewable sources and to meet peak demand \citep{zhao2023flexible, han2024hydropower, feng2022weekly}. In all such systems, dispatch and price formation rest on a quantity that is never observed: the opportunity cost of stored energy, referred to as the water value in hydro systems. It is an intertemporal shadow price, and it must be inferred from a forecast of uncertain future inflows. When that forecast is systematically optimistic, scarcity is priced as if it were less binding than it actually is, and every decision that depends on that price, such as reservoir operation, thermal commitment, spot-price formation, and forward contracting, inherits the distortion.

Where this shadow price is formed differs across institutional arrangements; the exposure to forecast quality does not. In cost-based systems, such as those prevalent across Latin America, the same forecasts feed a single optimization chain that determines both the dispatch and the price, so that a systematic forecast bias can contaminate operational and market signals at once. In market-based systems, by contrast, inflow forecasts inform decentralized bidding, and heterogeneous expectations are reflected in market participants' bids, which can partially offset individual forecast errors through the price formation process \citep{davilaparlatore2025identifying}. Even there, however, the same class of hydrothermal scheduling models is routinely used to produce the operational benchmarks and market-monitoring references against which observed behavior is assessed \citep{rangel2008competition, wolfgang2009hydro}, so that a common view of the future retains considerable influence over dispatch and prices.

Accordingly, the long-term management of scarce hydro resources remains a central challenge in countries such as Brazil, Chile, Colombia, Uruguay, and Peru, as well as in other regions with significant hydropower dependence, including parts of the United States, Canada, central Mexico, Norway, Vietnam, and New Zealand. More broadly, the mechanism is not specific to water. Battery storage is expanding rapidly as a flexibility resource for integrating variable renewable generation, with global installed capacity projected to reach approximately 900 GW by 2030 and 1,700 GW by 2035 \citep{IEA2025WEO}. As hydro reservoirs, short-duration storage, and long-duration storage technologies become increasingly important, biased expectations may similarly distort shadow values, intertemporal dispatch decisions, and market signals, potentially affecting investment decisions and the deployment of energy storage resources \citep{antweiler2021storage}.

Concerns about optimistic water-value assessments have been documented in Latin American power systems (e.g., \cite{ISCI2024}, Section~4.3), and a related body of scientific work has examined how modeling choices distort the opportunity cost of stored water. In particular, previous studies have shown that network simplifications \citep{rosemberg2021assessing} and simplified hydropower representations \citep{navarro2024medium} can affect the opportunity cost of stored water and lead to dispatch policies that perform poorly when implemented under more realistic operating conditions. In such cases, the model-implied scarcity signal may be too weak, requiring the system operator to correct the planned policy through out-of-merit dispatch to preserve reliability. Related distortions also arise in capacity expansion studies, where simplified representations of hydropower production functions can yield investment plans that perform poorly once uncertainty and operational detail are realized \citep{ramirez2019effect}. These concerns have motivated the development of energy scheduling and expansion models that explicitly pursue time-consistent formulations \citep{pisciella2016time}.

Despite the relevance of prior work, this literature addresses the \emph{structure} of the scheduling model (how the network, the plants, and the horizon are represented) and takes as given the inflow-forecasting process that feeds it. Forecasting inflows, however, is a challenge in itself. The literature on the topic is vast (see, e.g., \citealp{paula2015, erlon2024}). Machine learning and deep learning methods have improved short-term forecasts \citep{farfan2024streamflow, latif2023review, zhao2024comprehensive}. For the mid- and long-term horizons that matter for reservoir management, inflows depend on hard-to-predict climate variables, such as those associated with the El Ni\~no and La Ni\~na phenomena \citep{huang2019hydrological, velasquez2016modeling}. In addition, the convexity required by SDDP prevents the direct use of nonlinear forecasting models. For this reason, in Brazil and in other systems that rely on SDDP-based planning tools, inflows are commonly forecast with linear periodic autoregressive models \citep{pereira1984stochastic, noakes1985forecasting, matos2012computational, shapiro2013risk, GTMetodologiaCPAMP2022}. These models revert to the long-term mean of the historical record. Hence, they tend to produce optimistic forecasts when inflows remain below this mean for long periods. Such periods have been reported in Brazil, with severe and persistent droughts over the last decade \citep{geirinhas2021recent, marengo2022drought, de2024longest, w14040601}. The causes of these droughts are still under investigation and may include irrigation and land management, climate variability, and climate change \citep{chagas2022climate, mao2022phase, cook2020twenty}. In this context, \cite{brigatto2025assessing} show that the official Brazilian inflow forecasts exhibit a statistically significant optimistic bias for horizons from 1 to 24 months ahead over 14 years of data. These studies are concerned with the forecasts themselves, not with their consequences.

The knowledge gap addressed in this paper is therefore the following: no study has examined how a systematic bias in the inflow-forecasting process itself propagates through water values, reservoir trajectories, dispatch, price formation, and contracting incentives, nor quantified the operational, reliability, and market cost of that bias in a real system. Addressing this gap links the accuracy of inflow forecasts to their practical consequences for system operation and markets. Closing this gap requires three ingredients that rarely coexist: an analytical result linking forecast bias to the shadow value of storage; a system in which the bias has been previously documented over a long horizon; and public planning and operational data of sufficient quality. To the best of our knowledge, the Brazilian system is the case in which the three are simultaneously available, which is why it is used here as the empirical setting for a mechanism that is not particular to it.
\subsection{The Brazilian case: institutional setting, data, and documented evidence}

Brazil has one of the world’s largest installed hydropower capacities. As of December 2025, hydroelectric plants accounted for 57.7\% of the 79.9 avgGW (700 TWh) of electricity generated in 2025 \citep{ONS2025}. The country is one of the largest Latin American economies in terms of GDP and a global reference in hydro scheduling, and it occupies a distinctive position by combining a hydro-dominated generation mix with deep multi-year storage and a centralized stochastic dispatch framework in which operating decisions and spot prices are determined directly from the planning model (see \citealp{ribeiro2023technical}). As a result, inaccuracies in the water scheduling planning process can simultaneously propagate to reservoir operation, dispatch decisions, electricity prices, and market outcomes. The evidence reported below draws on official planning and operational records for the Brazilian interconnected system from January 2014 to October 2026; all data, processing scripts, and replication materials are publicly available \citep{Brigatto2026BiasedInflowsData}.

On the operational side, hydrothermal coordination is particularly complex in Brazil, where more than 160 hydroelectric plants with multi-year regulation reservoirs are arranged in interconnected cascades under different ownership structures. The resulting temporal coupling and the market externalities of long multi-owner cascades largely explain why the country's operation and short-term electricity market still rely on a cost-based, centralized scheduling framework. To support operational decisions, Brazil, as well as most hydro-dominated systems, relies on a chain of statistical and optimization models (see \cite{silva2025overview} for the Brazilian framework and \cite{fullner2025stochastic} for a survey of related methodologies). At the core of this chain is a hydro-scheduling planning model\footnote{In Brazil, the official hydro-scheduling planning model is implemented in the NEWAVE software.}, which typically applies the Stochastic Dual Dynamic Programming (SDDP) algorithm of \cite{pereira1991multi} to approximate the cost-to-go functions\footnote{Also known as recourse or value functions.} that guide the management of stored water through time. These cost-to-go functions quantify the opportunity cost of water by expressing the expected future operating cost as a function of current reservoir storage and hydrological conditions.

As mentioned above, the official Brazilian inflow forecasts are given by a periodic autoregressive model (PARp-A), and their optimistic bias is documented in \cite{brigatto2025assessing}. Using official data and a rolling-window statistical analysis, that study shows that the official implementation of the PARp-A model has systematically overpredicted inflows across all relevant forecasting horizons for more than a decade, revealing persistent and statistically significant biases. Whatever the causes of the low-inflow regime behind these biases, the biases themselves persist because water value computation is centralized and based on a single view of the future, as discussed above.

In this context, there are growing concerns about the effects of the reported bias on Brazilian supply adequacy and market incentives (see evidence and further references in \cite{brigatto2025assessing}). From 2013 to 2021, during the series of severe droughts mentioned above \citep{w14040601}, official inflow forecasts systematically overestimated observed inflows and the official dispatch methodology continued to postpone preventive thermoelectric generation, even as realized hydrological and storage conditions deteriorated. In practice, however, the system operator and the Brazilian Electric Sector Monitoring Committee\footnote{The Brazilian Electric Sector Monitoring Committee (CMSE) is a federal government body created by Law No.\,10,848/2004, responsible for monitoring supply adequacy, assessing operational risks, and recommending preventive and corrective measures.} repeatedly relied on emergency out-of-merit thermal dispatch to preserve supply adequacy. This divergence between the model-prescribed dispatch and the implemented operation is consistent with the operator assigning greater scarcity value to stored water than that implied by the official model. Although out-of-merit dispatch may also reflect other operational concerns, its persistence during a period of low realized inflows and optimistic forecasts suggests that the official water-value signal was insufficient to support the reliability actions ultimately deemed necessary. These emergency interventions imposed more than USD 3.8 billion in out-of-merit costs on consumers during the 2021 crisis and, within months, undermined the credibility of the official policy framework.\footnote{The Brazilian independent system operator reports the official cost assessment in Figure 33 of \cite{ONS_out_of_merit}.}

These concerns have already prompted regulatory responses. Since 2022, Brazilian authorities have advanced model- and data-governance reforms, culminating in National Energy Policy Council Resolution No. 1/2024 of March 12, 2024, which seeks greater coherence across the data, parameters, methodologies, and computational models used in sector scheduling planning and operation\footnote{See \cite{street2024modelgovernance} for a compilation of studies, official decisions, and an English translation of the resolution.}. Yet these reforms do not provide a systematic scientific assessment of how forecast bias propagates through operation, prices, and market incentives.
\subsection{Research questions and contributions}

Against this background, and given the central role of inflow forecasts in assessing water values in hydro-dependent power systems, the following two closely related questions remain unanswered: (i) Is an optimistically biased forecasting methodology capable of biasing water values, in the short- and long-term operation, yielding higher hydro discharges than those that would be obtained under unbiased forecasts? (ii) If affirmative, can the induced dispatch policy generate long-term operational and market distortions, ultimately affecting the system supply adequacy, spot-price formation, operating costs, and forward-contracting incentives?

\emph{Therefore, the objective of this paper is to extend the forecast bias findings in \cite{brigatto2025assessing} and provide formal and empirical evidence on this propagation channel.} To do that, first, we present a formal analytical result showing that, in a stylized setting, a positive (optimistic) inflow bias reduces estimated water values, thereby increasing water utilization, reducing storage, and postponing thermoelectric dispatch in comparison to the unbiased optimal case. Second, using official Brazilian operation planning model and data, we provide empirical evidence consistent with these effects in a system-wide setting. Third, we conduct controlled SDDP experiments to isolate and quantify the model-implied counterfactual effects of biased inflow forecasts on reservoir trajectories, thermal dispatch, spot prices, system risk, and operating costs. Fourth, we add a stylized forward-contracting overlay showing that biased forecasts can change the joint distribution of hydropower generation and spot prices, increase price--quantity risk, and reduce the contracting level of a representative risk-averse hydropower producer.

These four elements play distinct roles and support conclusions of different strength. The analytical result establishes the direction of the forecast bias effect in a stylized model. The historical evidence shows that the Brazilian operation record is consistent with the theoretical analysis, but it cannot establish causality, since realized operation is also affected by hydrological variability, operational adjustments, and other factors, and no unbiased counterfactual was observed. The controlled SDDP experiments provide the causal evidence within the model. The two policies differ only in the inflow process used for training and are evaluated under the same bias-corrected scenarios, so that the differences in the results can be attributed to the forecast bias. Finally, the contracting overlay indicates the incentive distortions of the mechanism in the forward market, under stated assumptions.

\section{Implications of Forecast Bias for the Cost-to-Go Function} \label{sec:implications}

In this section, we formally analyze how optimistic inflow forecasts distort water values and induce higher conditional first-stage hydro discharge decisions and, as a consequence, reduce second-stage storage levels relative to the unbiased optimal values. Because the system operator implements only the current-period (first-stage) dispatch decision, updates the storage level based on realized inflow and hydro discharge, and then repeats the process in the next period, in the stylized setting considered here, the sequence of biased dispatch decisions produces storage levels systematically below the unbiased optimal trajectory. This result provides the key analytical link between biased forecasts, depressed water values, delayed preventive thermal dispatch, and the broader effects discussed in Section \ref{sec:introduction}. 

Forecast bias is a statistical concept. In this section, however, we assume a deterministic setting for simplicity and didactic purposes. Therefore, a positive (optimistic) bias is modeled as a positive shift relative to a perfect forecast, which coincides with the realized inflow. Additionally, we consider a stylized version of the hydrothermal dispatch problem with a single hydropower plant, a set of thermal power plants, and no network constraints. 
At the end of this section, we discuss the limitations of this analysis and the extent to which the insights obtained here are valid. We will also discuss some intuition regarding the results derived here and how they extend to practical cases.

For now, suppose the ISO determines the optimal first-stage dispatch decisions, given by the thermal generation vector ($\boldsymbol{g}_1$), hydro discharge by the hydro unit to generate energy ($u_1$), and the hydro spillage ($s_1$), by solving a deterministic multistage optimization problem. The look-ahead or planning component, which comprises a set of advisory future dispatch decisions ($\boldsymbol{g}_t, u_t, s_t$) for stages $t = 2, \ldots, T$, is used to capture the opportunity cost of water observing the future demand ($d_t$), the water inflows ($y_t$), the vector of marginal costs for the thermal units ($\boldsymbol{c}_t$), storage levels ($v_t$), and the system's constraints.\footnote{As shown in \cite{pereira1991multi}, the planning component can be efficiently approximated via SDDP through a recursive procedure. In this section, however, we use an explicit representation of the cost-to-go function to isolate the effect of bias and leave technical details about SDDP to be discussed in \ref{sec:appendix_sddp}.} Assume also that, for each period $t \in \{1,\dots,T\}$, the ISO is given a deterministic inflow forecast $y_t$. The resulting optimization problem used to define the first-stage dispatch (thermal generation) is given by:
{\allowdisplaybreaks
\begin{align}
&\underset{\substack{v_t, \boldsymbol{g}_t, u_t, s_t
}
}{\min}\;  
\sum_{t=1}^T\boldsymbol{c}^\top_{t} \boldsymbol{g}_t 
\label{eq:objective_hypothesis} \\
\text{\emph{subject to}}
&\quad u_t + \boldsymbol{1}^\top\boldsymbol{g}_t = d_{t},\quad \forall t \in \{1,\dots,T\} \label{eq:energy_balance_hypothesis} \\
&\quad v_{t} = v_{t-1} + y_t - u_t - s_t\label{eq:water_balance_hypothesis},\quad \forall t \in \{1,\dots,T\}  \\
&\quad \underline{V} \leq v_t \leq \overline{V},\quad \forall t \in \{1,\dots,T\} \label{eq:reservoir_bounds_hypothesis} \\
&\quad \underline{\boldsymbol{G}} \leq \boldsymbol{g}_t \leq \overline{\boldsymbol{G}},\quad \forall t \in \{1,\dots,T\} \label{eq:thermal_bounds_hypothesis} \\
&\quad(u_t, s_t) \in \mathcal{P},\quad \forall t \in \{1,\dots,T\}. \label{eq:hydro_bounds_hypothesis}
\end{align}}
For each period $t$, in \eqref{eq:objective_hypothesis}-\eqref{eq:hydro_bounds_hypothesis}, constraints \eqref{eq:energy_balance_hypothesis} and \eqref{eq:water_balance_hypothesis} represent the energy and water balance equations, respectively. Constraints \eqref{eq:reservoir_bounds_hypothesis}-\eqref{eq:hydro_bounds_hypothesis} impose bounds on the decision variables and polyhedral constraints, $(u_t, s_t) \in \mathcal{P}$, to model hydro discharge and spillage joint and individual upper limits. For the analysis that follows, we take the free-spill specialization $\mathcal{P}=\{(u_t,s_t): 0\le u_t\le\overline U,\ s_t\ge 0\}$, under which spillage disposes of any surplus inflow at zero cost.

Now, suppose the ISO is given an optimistically biased inflow forecast $y_t + \delta_t$, where $\delta_t \geq 0$ represents a bias from the second stage onward for each period. In this scenario, the future water availability is overestimated by a vector $\boldsymbol{\delta}=(\delta_2,\dots,\delta_T)\geq0$. In order to study the effect of these biases on the first-stage decisions, it is useful to analyze the problem in its dynamic programming version. 

For each period $t$, given the storage state $v_{t-1}$ 
and a forecast bias vector $\boldsymbol{\delta} = (\delta_2, \dots, \delta_T) 
\geq \boldsymbol{0}$, the cost-to-go function is defined recursively as:
\begin{align}
Q_t^{(\boldsymbol{\delta})}(v_{t-1})
=
\min_{r_t,v_t}
&\quad
\phi_t(r_t)+Q_{t+1}^{(\boldsymbol{\delta})}(v_t)
\label{eq:reduced_stage_problem_obj}\\
\text{\emph{subject to}}
&\quad
v_t=v_{t-1}+y_t+\delta_t-r_t,
\label{eq:reduced_stage_problem_balance}\\
&\quad
\underline V\le v_t\le \overline V.
\label{eq:reduced_stage_problem_storage_bounds}
\end{align}
Where $\phi_t$, the immediate cost of period $t$ written as a function of the total water release $r_t$, is given by
\begin{align}
\phi_t(r_t)
:=
\min_{\boldsymbol g_t, u_t}
&\quad
\boldsymbol c_t^\top \boldsymbol g_t
\label{eq:phi_definition_obj}\\
\text{\emph{subject to}}
&\quad
u_t+\boldsymbol 1^\top\boldsymbol g_t=d_t,
\label{eq:phi_definition_balance}\\
&\quad
\underline{\boldsymbol{G}} \le \boldsymbol g_t\le \overline{\boldsymbol G}
\label{eq:phi_definition_bounds}\\
&\quad
u_t \leq r_t
\label{eq:phi_definition_release}\\
&\quad
0\le u_t\le \overline U,
\label{eq:phi_definition_turbine_bound}
\end{align}
where the slack of constraint \eqref{eq:phi_definition_release} is the hydro spillage variable. Thus, \(\phi_t\) is piecewise-linear convex and non-increasing in \(r_t\), as the more water released in period $t$, the lower the thermal dispatch cost in that period. We fix the boundary condition $Q_{T+1}^{(\boldsymbol{\delta})}(v_T) = 0$ 
for all $v_T$, and set $\delta_1 = 0$ by convention (no bias at the 
first stage). Each $Q_t^{(\boldsymbol{\delta})}$ is convex and 
non-increasing in $v_{t-1}$, with subdifferential $\partial_v 
Q_t^{(\boldsymbol{\delta})} \leq 0$. Subgradients in $\partial_v 
Q_t^{(\boldsymbol{\delta})}$ play a central role in hydrothermal power system operation planning because they measure the reduction in future operating costs associated with an additional unit of water stored. Interpreting each element of this set as an opportunity cost of using water today, commonly referred to as a water value, we define:
\begin{definition}
The water-value correspondence is defined as the negative of the cost-to-go function subdifferential with respect to storage: $W_t^{(\boldsymbol{\delta})}(v_{t-1})
=
-\partial_v Q_t^{(\boldsymbol{\delta})}(v_{t-1})$.
Each element $w_t \in W_t^{(\boldsymbol{\delta})}(v_{t-1})$ is a feasible marginal opportunity cost of stored water at state $v_{t-1}$.
\end{definition}

\subsection{Positive forecast bias induces higher water use}

We prove this result in two steps. First, Proposition \ref{prop:bias_reduces_water_value} shows that, regardless of any existing bias from $t=t_0+1$ to $T$, with $t_0\in\{2,...,T\}$, an additional non-negative bias at period $t_0$ weakly reduces the water values in all periods $t=2,...,t_0$ and, consequently, weakly increases the optimal hydro discharge $u_t$ from $t=1,..., t_0$. The term weakly indicates that equality is allowed. Thus, weakly reduces means “reduces or remains unchanged,” and weakly increases means “increases or remains unchanged.” Then, Theorem~\ref{theo:bias_increase_discharge} applies this result recursively to conclude that the optimal first-stage hydro discharge is weakly larger under any optimistic bias than under unbiased forecasts.

\begin{proposition}[marginal bias effect]
\label{prop:bias_reduces_water_value}
Let  $\boldsymbol{\delta} = (0,\dots,0, \allowbreak\delta_{t_0+1},\dots,\delta_{T}) \geq \boldsymbol{0}$, with zeros in periods $1,\dots,t_0$, and let
$\boldsymbol{\delta}'=\boldsymbol{\delta}+\eta\,\boldsymbol{e}_{t_0}$, for some $t_0\in\{2,\dots,T\}$ and $\eta\geq 0$. Then:
\begin{enumerate}
\item[(i)] For every $v_{t-1}\in[\underline V,\overline V]$ and every $t=2,\dots,t_0$,
\begin{align}
\partial_v Q_t^{(\boldsymbol{\delta}')}(v_{t-1})
\;\succeq\;
\partial_v Q_t^{(\boldsymbol{\delta})}(v_{t-1}),
\label{eq:prop_subgrad_ordering}
\end{align}
where $\succeq$ denotes endpoint ordering of subdifferential intervals.\footnote{%
For nonempty closed intervals $A,B\subseteq\overline{\mathbb R}$, we write $A\succeq B$ iff
$\inf A\ge \inf B$ and $\sup A\ge \sup B$. In other words, the interval $A$ lies weakly to the right of $B$ in
the endpoint sense; reading the endpoints as $\inf/\sup$ in the extended reals accommodates the half-line subdifferentials that arise at the storage bounds.}
\item[(ii)] Among the optimal solutions of each stage-$t$ problem the discharges form a closed interval; taking $u_t^{*}$ to be its greatest or lowest element, for any given $t \in \{1, ...,t_0\}$ and state $v_{t-1}\in[\underline V,\overline V]$,
\begin{align}
u_t^{*(\boldsymbol{\delta}')}\;\geq\;u_t^{*(\boldsymbol{\delta})}. 
\label{eq:prop_u_ordering}
\end{align}
\item[(iii)] Consequently, for any common initial state $v_{t-1}$, the corresponding (lowest or greatest) optimal storage is such that
\begin{align}
v_t^{*(\boldsymbol{\delta}',\boldsymbol{\delta})}\;\leq\;v_t^{*(\boldsymbol{\delta},\boldsymbol{\delta})}, \label{eq:prop_v_ordering}
\end{align}
where the first superscript, $\boldsymbol{\delta}'$, denotes the bias vector under which the water release was defined (giving rise to the implemented discharge $u_t^{*(\boldsymbol{\delta}')}$ of part~(ii)), and the second, $\boldsymbol{\delta}$, the bias vector affecting the stage-$t$ inflow, i.e., $v_t^{*(\boldsymbol{\delta}',\boldsymbol{\delta})}=v_{t-1}+y_t+\delta_t-r_t^{*(\boldsymbol{\delta}')}$.
\end{enumerate}
\end{proposition}
\noindent Note that this inequality is valid for a given period $t$ and initial state $v_{t-1}$. Therefore, inequalities \eqref{eq:prop_u_ordering} and \eqref{eq:prop_v_ordering} should be read respectively as $u_t^{*(\boldsymbol{\delta}')}(v_{t-1})\;\geq\;u_t^{*(\boldsymbol{\delta})}(v_{t-1})$ and $v_t^{*(\boldsymbol{\delta}',\boldsymbol{\delta})}(v_{t-1})\;\leq\;v_t^{*(\boldsymbol{\delta},\boldsymbol{\delta})}(v_{t-1})$. To keep the notation light, we omit the argument $v_{t-1}$ from all affected decision variables in what follows.

The complete mathematical proof is given in~\ref{appendix:proofs}. Broadly, the argument therein is based on a backward induction on $t$, from $t=t_0 \geq 2$ down to $t=1$. At the base case $t=t_0$, the bias enters the water balance as additional inflow at exactly that period; substituting in the recursion shows that the $\boldsymbol{\delta}'$ problem at state $v_{t_0-1}$ coincides with the $\boldsymbol{\delta}$ problem evaluated at the shifted state $v_{t_0-1}+\eta$. Convexity of $Q_{t_0}^{(\boldsymbol{\delta})}$ then translates this shift identity into the endpoint-subgradient ordering~\eqref{eq:prop_subgrad_ordering} at $t=t_0$. The induction step combines two ingredients. First, we use the fact that the difference between $Q^{(\boldsymbol{\delta}')}_{t+1}$ and $Q^{(\boldsymbol{\delta})}_{t+1}$ is non-decreasing in $v_t$ to deliver the hydro-discharge ordering \eqref{eq:prop_u_ordering}: under $\delta'$, an optimal post-state $v_t$ can be selected no greater than under $\delta$, hence the corresponding total water release $r_t$ is no smaller, and the inner subproblem~\eqref{eq:phi_definition_obj}--\eqref{eq:phi_definition_turbine_bound} allocates the increment to $u_t$ before spillage. Second, a KKT envelope-theorem case-split on the storage-bound active set transmits the subgradient ordering from stage $t+1$ to stage $t$ showing that $\partial_v Q_t^{(\boldsymbol{\delta}')}(v_{t-1})
\;\succeq\;
\partial_v Q_t^{(\boldsymbol{\delta})}(v_{t-1})$. This endpoint ordering implies that the difference between $Q_t^{(\boldsymbol{\delta}')}$ and $Q_t^{(\boldsymbol{\delta})}$ is non-decreasing in $v_{t-1}$, which closes the induction step and provides the monotonicity property needed to repeat the same argument one period earlier.

Figure~\ref{fig:shift_illustration} illustrates the core of this argument for a stage $t=t_0-1$, in which the two problems differ only in the cost-to-go function. Panel (a) shows the immediate cost written as a function of the end-of-stage storage, $\phi_t(v_{t-1}+y_t-v_t)$, the cost-to-go function under $\boldsymbol{\delta}$, and their sum, whose minimizer defines the optimal end-of-stage storage. In panel (b), the additional bias $\eta$ at stage $t+1$ shifts the cost-to-go function to the left, as stated by the shift identity. The optimal end-of-stage storage moves to the left and, therefore, the water release increases. Panel (c) shows that the additional bias leads to a weakly lower water value at every storage level.

\begin{figure}[!htbp]
\centering
\begin{minipage}{0.32\textwidth}
\centering
\begin{tikzpicture}[x=0.42cm,y=0.24cm,>=stealth,font=\footnotesize]
\draw[->] (0,0)--(10.9,0) node[below left]{$v_t$};
\draw[->] (0,0)--(0,18.6) node[above]{Cost};
\draw[blue!70!black,thick] (0,16)--(3,7)--(7,3)--(10,3);
\draw[green!45!black,thick,dashed] (0,1)--(2,1)--(6.5,3.25)--(10,10.25);
\draw[black,very thick] (0,17)--(2,11)--(3,8.5)--(6.5,6.75)--(7,7.25)--(10,13.25);
\draw[densely dotted] (6.5,6.75)--(6.5,0) node[below]{$v_t^{*(\boldsymbol{\delta})}$};
\fill (6.5,6.75) circle (1.6pt);
\node[blue!70!black,anchor=south] at (8.6,3) {$Q_{t+1}^{(\boldsymbol{\delta})}$};
\node[green!45!black,anchor=north west] at (8.9,7.6) {$\phi_t$};
\node[anchor=south west] at (3.2,9.4) {$\phi_t+Q_{t+1}^{(\boldsymbol{\delta})}$};
\end{tikzpicture}

{\footnotesize (a) Bias vector $\boldsymbol{\delta}$}
\end{minipage}
\hfill
\begin{minipage}{0.32\textwidth}
\centering
\begin{tikzpicture}[x=0.42cm,y=0.24cm,>=stealth,font=\footnotesize]
\draw[->] (0,0)--(10.9,0) node[below left]{$v_t$};
\draw[->] (0,0)--(0,18.6) node[above]{Cost};
\draw[gray!55,thick] (0,16)--(3,7)--(7,3)--(10,3);
\draw[gray!55,thick] (0,17)--(2,11)--(3,8.5)--(6.5,6.75)--(7,7.25)--(10,13.25);
\draw[blue!70!black,thick] (0,10)--(1,7)--(5,3)--(10,3);
\draw[green!45!black,thick,dashed] (0,1)--(2,1)--(6.5,3.25)--(10,10.25);
\draw[black,very thick] (0,11)--(1,8)--(2,7)--(5,5.5)--(6.5,6.25)--(10,13.25);
\draw[gray!70,densely dotted] (6.5,6.75)--(6.5,0);
\fill[gray!70] (6.5,6.75) circle (1.4pt);
\node[gray!90,anchor=north west] at (6.1,0) {$v_t^{*(\boldsymbol{\delta})}$};
\draw[densely dotted] (5,5.5)--(5,0);
\fill (5,5.5) circle (1.6pt);
\node[anchor=north east] at (5.4,0) {$v_t^{*(\boldsymbol{\delta}')}$};
\draw[gray!70,densely dotted] (7,3)--(7,0);
\draw[->,very thick,green!40!black] (7,1.4)--(5,1.4);
\node[green!40!black,anchor=west] at (7.1,1.4) {$\eta$};
\node[blue!70!black,anchor=south] at (8.6,3) {$Q_{t+1}^{(\boldsymbol{\delta}')}$};
\node[green!45!black,anchor=north west] at (8.9,7.6) {$\phi_t$};
\node[anchor=south west] at (2.5,11.6) {$\phi_t+Q_{t+1}^{(\boldsymbol{\delta}')}$};
\end{tikzpicture}

{\footnotesize (b) Bias vector $\boldsymbol{\delta}'=\boldsymbol{\delta}+\eta\,\boldsymbol{e}_{t+1}$}
\end{minipage}
\hfill
\begin{minipage}{0.32\textwidth}
\centering
\begin{tikzpicture}[x=0.42cm,y=1.2cm,>=stealth,font=\footnotesize]
\draw[->] (0,0)--(10.9,0) node[below left]{$v_t$};
\draw[->] (0,0)--(0,3.72) node[above]{Water value};
\draw[gray!55,thick] (0,3)--(3,3)--(3,1)--(7,1)--(7,0)--(10,0);
\draw[blue!70!black,thick] (0,3)--(1,3)--(1,1)--(5,1)--(5,0)--(10,0);
\node[gray!90,anchor=south west] at (3.1,1.9) {$W_{t+1}^{(\boldsymbol{\delta})}$};
\node[blue!70!black,anchor=north west] at (1.1,0.95) {$W_{t+1}^{(\boldsymbol{\delta}')}$};
\node[below] at (6.5,0) {\phantom{$v_t^{*(\boldsymbol{\delta})}$}};
\end{tikzpicture}

{\footnotesize (c) Water values at stage $t+1$}
\end{minipage}
\caption{Illustration of the effect of an additional bias $\eta$ at stage $t+1$ on the stage-$t$ problem. The cost-to-go function is represented by three linear segments. In panels (b) and (c), the curves in gray reproduce the case of panel (a).}
\label{fig:shift_illustration}
\end{figure}

Having established Proposition \ref{prop:bias_reduces_water_value}, we are now ready to state our main result.

\begin{theorem}[Biased forecasts weakly increase first-stage hydro discharge]
\label{theo:bias_increase_discharge}
For any nonnegative forecast bias vector $\boldsymbol{\delta}=(\delta_2,\dots,\delta_T)\geq\boldsymbol{0}$ (with $\delta_1=0$ by convention), among the optimal solutions for the first stage problem, the discharges form a closed interval; taking $u_1^{*}$ to be its greatest or lowest element, we have: $u_1^{*(\boldsymbol{\delta})}\;\geq\;u_1^{*(\boldsymbol{0})}$.
\end{theorem}

The complete mathematical proof is given in~\ref{appendix:proofs}. However, the main idea therein is to build a sequence of bias vectors that adds the components of $\boldsymbol{\delta}$ one at a time, from period $T$ backward to period $2$. Each consecutive pair satisfies the hypothesis of Proposition~\ref{prop:bias_reduces_water_value}; chaining part~(ii) at $t=1$ across the sequence yields the claim. 

A direct consequence of Theorem \ref{theo:bias_increase_discharge} is the following: since the end-of-period storage satisfies $v_1 = v_0 + y_1 - u_1 - s_1$ (with $\delta_1 = 0$), the biased policy yields $v_1^{*(\boldsymbol{\delta},\boldsymbol{0})} \leq v_1^{*(\boldsymbol{0},\boldsymbol{0})}$, and total first-stage thermal dispatch satisfies $\boldsymbol{1}^\top \boldsymbol{g}_1^{*(\boldsymbol{\delta})} \leq \boldsymbol{1}^\top \boldsymbol{g}_1^{*(\boldsymbol{0})}$ by the energy balance~\eqref{eq:energy_balance_hypothesis}; by the merit-order optimality of the stage-1 subproblem, the reduction displaces the most expensive units first. Repeated across periods in a rolling-horizon implementation, the conditional over-discharge produces a sequence of storage levels systematically weakly below the unbiased optimal trajectory.

\paragraph{\textbf{Remark 1:} Spot prices and the water value}
The spot price at stage $t$, denoted $\pi_t$, is the dual variable of the energy balance constraint~\eqref{eq:phi_definition_balance} of the inner subproblem~\eqref{eq:phi_definition_obj}--\eqref{eq:phi_definition_turbine_bound}. By linear-programming duality, $\pi_t$ measures the marginal increase in dispatch cost from a one-unit increase in demand $d_t$ and equals, in the typical case, the marginal cost of the most expensive thermal unit dispatched in merit order.

\medskip
\noindent Reading the stage-$t$ problem as a function of $u_t$ by eliminating $v_t$ via the water balance~\eqref{eq:reduced_stage_problem_balance}, the KKT optimality condition with respect to $u_t$ at an interior optimum $u_t^*\in(0,\overline U)$ yields $\pi_t \;=\; w_t$, where $w_t \in -\partial_v Q_{t+1}(v_t^*)$ is the water value. The water value plays the role of the marginal cost of hydropower in the merit order, exactly as if it were the unit's fuel cost. When hydro generation affected by the bias reaches its upper bound, then $\pi_t \ge w_t$. The reduction in water values may cause a merit-order displacement of expensive thermal generation, thereby pushing $\pi_t$ downward.

\medskip
\noindent It then follows from Theorem~\ref{theo:bias_increase_discharge} that an optimistic inflow forecast bias weakly reduces the spot price at every stage \emph{of the planning trajectory}. The reduction is artificial: it does not reflect any change in the system operative reality, but rather the price signal of a water valuation constructed from inflow forecasts that are not aligned with the true realized inflow process. In a market that uses these prices to settle dispatch, consumers receive a temporary discount paid for by depleted reservoirs and elevated reliability risk in subsequent periods. However, it is important to emphasize that, for periods $t>1$, this effect is valid only at the planning stage. As will be shown in the case studies, the trajectories resulting from lower reservoir levels may cause realized spot prices under the biased policy to significantly exceed, on average, those under the unbiased policy.

\paragraph{\textbf{Remark 2:} Insights on the effects of the bias on first-stage hydro discharges}
Three complementary insights reinforce and qualify the result of
Theorem~\ref{theo:bias_increase_discharge}. First, the backward induction in the proof of
Proposition~\ref{prop:bias_reduces_water_value} shows how a bias at period $t_0$ propagates backward through the merit-order channel: it shifts $\partial_v Q_{t_0}^{(\boldsymbol{\delta})}$ as in \eqref{eq:prop_subgrad_ordering}, displaces the marginal thermal unit at $t_0-1$, and thereby shifts $\partial_v Q_{t_0-1}^{(\boldsymbol{\delta})}$, and so on down to the first stage. Each backward step preserves the effect when storage is scarce and the water value is high, and absorbs it when storage is abundant and $\partial_v Q_t^{(\boldsymbol{\delta})}$ is already near zero.

\medskip
\noindent Second, the same mechanism implies that the magnitude of the effect identified
in Theorem~\ref{theo:bias_increase_discharge} is \emph{heterogeneous across
hydrological conditions}. In wet periods, $Q_t^{(\boldsymbol{\delta})}$
is approximately flat in $v_{t-1}$ over the relevant operating range
because water is abundant, so the gap $|\partial_v Q_t^{(\boldsymbol{\delta})} -
\partial_v Q_t^{(\boldsymbol{0})}|$ is small. In dry
periods, $Q_t^{(\boldsymbol{\delta})}$ has steep negative slope (the
marginal value of water is high), the bias produces a substantial
reduction in this slope, and the larger
water-value gap translates into a larger upward distortion in
first-stage hydro discharge.

\medskip
\noindent Third, combining both observations yields a policy-relevant
conclusion: the misallocation of water induced by optimistic forecast
bias is most severe precisely when water is scarcest. In dry periods,
each backward step preserves the effect of the bias (first insight)
and the water-value gap it creates is largest (second insight), so an
optimistic bias is most likely to produce reservoir drawdown
trajectories that deviate significantly from the unbiased optimum, with
the largest consequences for system reliability. This asymmetry
suggests that forecast-bias monitoring should be prioritized during dry
seasons or drought episodes, precisely when the operational stakes are
highest.

\paragraph{\textbf{Remark 3:} Conditional versus dynamic effects}
Theorem~\ref{theo:bias_increase_discharge} is a conditional statement: holding the
initial storage and the planning horizon fixed, an optimistic bias
induces $u_1^{*(\boldsymbol{\delta})} \geq u_1^{*(\boldsymbol{0})}$.
More generally, Step~1 of the induction, applied along the same one-component-at-a-time argument used in  Theorem~\ref{theo:bias_increase_discharge} at any shared state $v_{t-1}$, establishes that, for
every $t = 1,\dots,T$,
$u_t^{*(\boldsymbol{\delta})} \geq u_t^{*(\boldsymbol{0})}$
at every shared state $v_{t-1}$. In a rolling-horizon implementation, however, the over-discharge at $t=1$ depletes storage, so the biased policy enters period $t=2$ from a lower state $v_1^{*(\boldsymbol{\delta},\boldsymbol{0})} \leq v_1^{*(\boldsymbol{0},\boldsymbol{0})}$. The conditional comparison at $t=2$ no longer applies between the two trajectories, since they are evaluated at different states. After some periods, the biased policy may, for the same period $t$, dispatch \emph{less} hydro and more thermal in absolute terms, because its deteriorated reservoir state pushes the conditional water value in 
$|\partial_v Q_t^{(\boldsymbol{\delta})}(v_{t-1}^{*(\boldsymbol{\delta},\boldsymbol{0})})|$ above the value visited by the unbiased policy. This is the expected dynamic feedback. On a dry horizon, the biased trajectory therefore incurs higher integrated thermal cost (the biased decisions are
suboptimal under the true cost-to-go) and faces an increased risk of storage exhaustion before the next wet season. The conditional result of Theorem~\ref{theo:bias_increase_discharge} is therefore a necessary piece for understanding the tendency to deplete reservoirs, especially in dry periods.

It is worth mentioning that the assumptions made in this section (e.g., deterministic inflows, a single reservoir, and no network constraints) can be deemed unrealistic for real power systems. However, it is important to emphasize that most long-term hydrothermal planning models do rely on many simplifications; for example, the official Brazilian software that computes water values still does not consider network constraints explicitly. Notwithstanding, the objective of this section is not to affirm that this effect will always be present, but rather to demonstrate the forces acting on water usage and reservoir levels in an ideal case. The insights drawn here are then used to support the empirical evidence presented in the next sections.

\subsection{Extension to the stochastic setting}

The deterministic argument can be extended directly to the stagewise-independent stochastic multistage setting. In this paper, we use the same model implemented in \cite{street2020assessing}. Conditional on each inflow innovation, the stage problem is deterministic and its cost-to-go is convex, so the argument of Proposition~\ref{prop:bias_reduces_water_value} applies at each stage; taking conditional expectations stage by stage carries the ordering~$\succeq$ to the expected cost-to-go.

\medskip

\noindent In practice, however, the stochastic model used for the inflows follows a periodic autoregressive form, i.e., $y_t = \boldsymbol{M}_t \boldsymbol{y}_{[t-1]}
      + \boldsymbol{N}_t \boldsymbol{\epsilon}_t,$ where $\boldsymbol{y}_{[t-1]}=(y_1,\dots,y_{t-1})$ denotes the inflow history up to stage $t-1$, and $\boldsymbol{\epsilon}_t$ is a random
innovation with known distribution \citep{brigatto2025assessing} and zero mean (in the unbiased case). 

In this setting, a bias in a given period $t$ may affect subsequent periods in different ways, depending on the sign of matrix $\boldsymbol{M}_t$. Thus, some additional assumptions would be needed to ensure a resulting positive bias effect, such as the non-negativeness of this matrix, or on the composite effect of the bias on $t$ and its propagation on subsequent cost-to-go functions. Because the objective of this section is to shed light on the main effect of a positive bias on the operation of a hydrothermal power system, rather than to prove that this effect is always present in every possible stochastic setting, we restrict the analytical proof to the canonical stagewise-independent framework of \cite{pereira1991multi}.\footnote{This setting allows the effect of the bias on water values, dispatch decisions, and cost-to-go functions to be isolated and understood in a transparent and pedagogical manner. The more realistic periodic autoregressive model introduces propagation effects that depend on the structure of the matrices $\boldsymbol{M}_t$, potentially amplifying, attenuating, or even reversing the direct impact of the bias. Therefore, rather than imposing additional assumptions to recover a general theoretical result, we use the empirical and controlled case studies presented in the following sections to assess the magnitude and direction of the bias effect under the realistic stochastic processes currently adopted in practice.} 

In this setting, a forecast bias is modeled as a shift $\delta_t\ge 0$ added to the inflow forecast at each future stage $t\ge 2$. The stochastic cost-to-go function is then $Q_t^{(\boldsymbol{\delta})}(v_{t-1})
=
\mathbb E_{\epsilon_t}
\left[
q_t^{(\boldsymbol{\delta})}
(v_{t-1},\hat{y}_t +  \epsilon_t)
\right],$ where \(q_t^{(\boldsymbol{\delta})}\) denotes the stage-\(t\) recourse
problem conditional on \(\epsilon_t\) and $\hat{y}_t$ is the expected value of $y_t$.

\begin{theorem}[Stochastic extension]
\label{theo:theorem_stochastic}
For any forecast bias vector $\boldsymbol{\delta}=(\delta_2,\dots,\delta_T)\ge\boldsymbol 0$, and provided the inflow process is almost surely nonnegative and the standing assumptions of \ref{appendix:proofs} hold, $\partial_v Q_t^{(\boldsymbol{\delta})}
(v_{t-1})
\;\succeq\;
\partial_v Q_t^{(\boldsymbol 0)}
(v_{t-1})$ 
\; for all $v_{t-1}$ in the support of the state process,\; for $t=2,\dots,T$. Consequently, taking $u_1^{*}$ to be the greatest or lowest optimal solution for $u_1$, we have: $u_1^{*(\boldsymbol{\delta})}\;\ge\;u_1^{*(\boldsymbol 0)}$.

\end{theorem}

\noindent The complete mathematical proof is given in~\ref{appendix:proofs}. %

\paragraph{\textbf{Remark 4:} Structural incentive, stochastic masking, and crisis amplification}
The structural effect of Theorems~\ref{theo:bias_increase_discharge} and~\ref{theo:theorem_stochastic} should not be read as implying that the biased policy dispatches more hydro than the unbiased policy along every realized rolling-horizon trajectory: two distinct mechanisms separate the local incentive from realized outcomes.

\medskip
\noindent (i) \emph{Trajectory divergence.} The biased policy depletes storage more aggressively in early periods, so the trajectories under the biased and unbiased policies diverge over time and future decisions are evaluated at different realized states. After sufficient depletion, the biased trajectory may reach states in which hydro availability is so scarce that realized hydro dispatch falls below the unbiased trajectory, even though the biased policy remains systematically more optimistic in its water valuation and, thereby, more aggressive in its one-stage conditional water usage.

\medskip
\noindent (ii) \emph{Stochastic masking.} Realized operating costs depend on the realized inflow path. Wet realizations may temporarily offset the effects of overly aggressive water use, so along such paths the biased policy can appear ex post well calibrated and even outperform the unbiased policy that conserved water. This is not evidence that the bias was correct; it is the natural consequence of a wrong forecast happening to align with a favorable realization. The masking effect is asymmetric: by Remark~2, in wet realizations the bias effect is small, whereas in dry realizations the mechanisms of Remark~3 (early depletion, higher integrated thermal cost, risk of storage exhaustion) combine to produce sharp realized costs and high spot prices. A sequence of wet years can therefore validate empirically a forecast that is systematically optimistic, building operator and regulator confidence precisely until a dry realization makes the accumulated misallocation binding and irreversible \citep{brigatto2025assessing}.

\medskip
\noindent In practice, what matters is whether the cumulative effect of forecast bias can materialize in actual system operation, leading reservoir levels to fall below planned trajectories, thermal generation to be postponed relative to what later proves necessary, and corrective actions to become more frequent and severe due to the deteriorated system storage conditions. In the next section, we discuss how traces of this cumulative effect can be identified in empirical data from the Brazilian system operation.

\section{Empirical evidence of biased look-ahead in the Brazilian power system operation planning} \label{sec:empirical_evidence}

We now turn from the analytical mechanism of Section \ref{sec:implications} to empirical evidence from the Brazilian power system. The purpose of this section is to examine whether the historical planning and operation data, from January 2014 to October 2026, display patterns consistent with the forecast-bias channel identified in the previous section. Since the Brazilian ISO has historically relied on the same official forecasting framework, there is no observed real-world counterpart in which operation was planned and implemented under an unbiased inflow-forecasting methodology. The evidence in this section should therefore be interpreted as diagnostic: it assesses whether forecast errors, planned reservoir trajectories, implemented dispatch, spot prices, and out-of-merit actions move in a way consistent with biased look-ahead. The controlled counterfactual comparison is developed in Section \ref{sec:case_study}. We begin with a brief overview of the Brazilian generation mix and storage structure, because the magnitude of the effects documented below depends on the scale of hydropower and reservoir storage in the system. All data, processing scripts, and replication materials used in this Section are available in \cite{Brigatto2026BiasedInflowsData}.

In recent years, the Brazilian power system, like many other energy systems worldwide, has experienced increasing participation from wind, solar, and distributed generation. Nevertheless, when each source is considered separately, hydropower remains the country's largest source of electricity supply, as shown in Figures \ref{fig:empirical_evidence_generation_share} and \ref{fig:empirical_evidence_generation_avg} (see \cite{ONS2025}).

\begin{figure}[!ht]
\centering
\begin{minipage}{0.48\textwidth}
    \centering
    \includegraphics[
        width=\textwidth,
        trim={0.0cm 0.0cm 0.0cm 0.0cm},
        clip
    ]{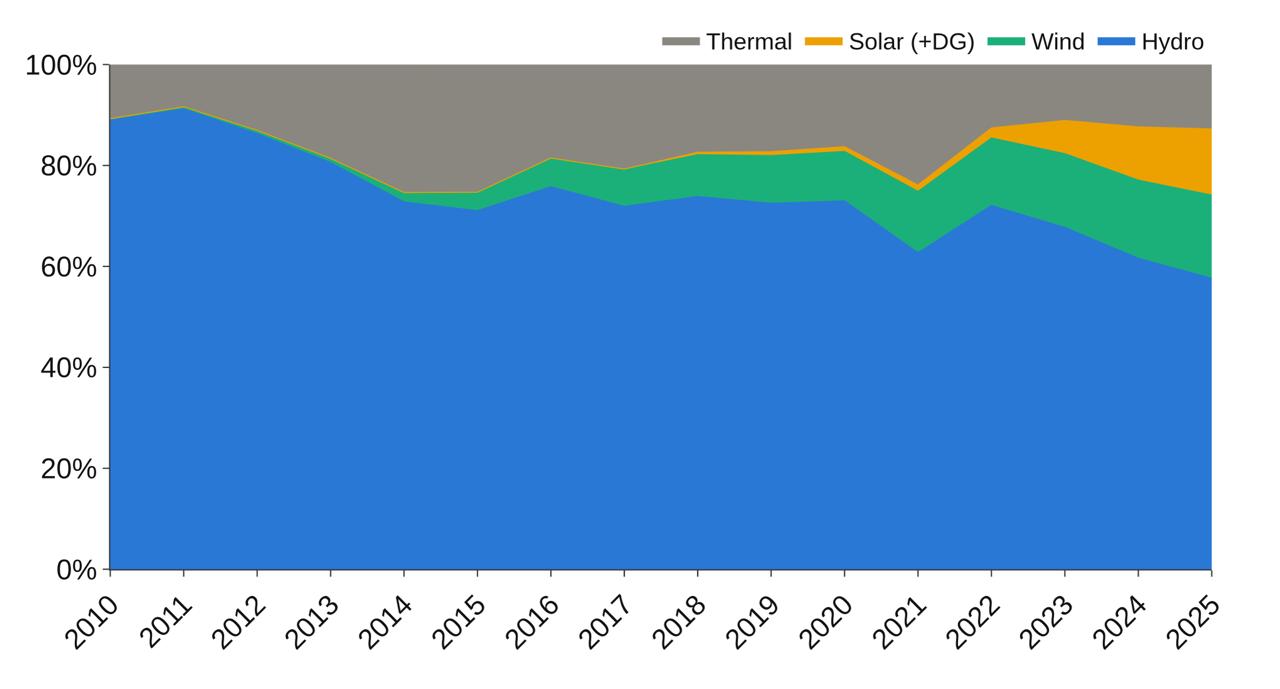}
    \captionof{figure}{Evolution of annual electricity generation shares in the Brazilian power system by source.}
    \label{fig:empirical_evidence_generation_share}
\end{minipage}
\hfill
\begin{minipage}{0.48\textwidth}
    \centering
    \includegraphics[
        width=\textwidth,
        trim={0.0cm 0.0cm 0.0cm 0.0cm},
        clip
    ]{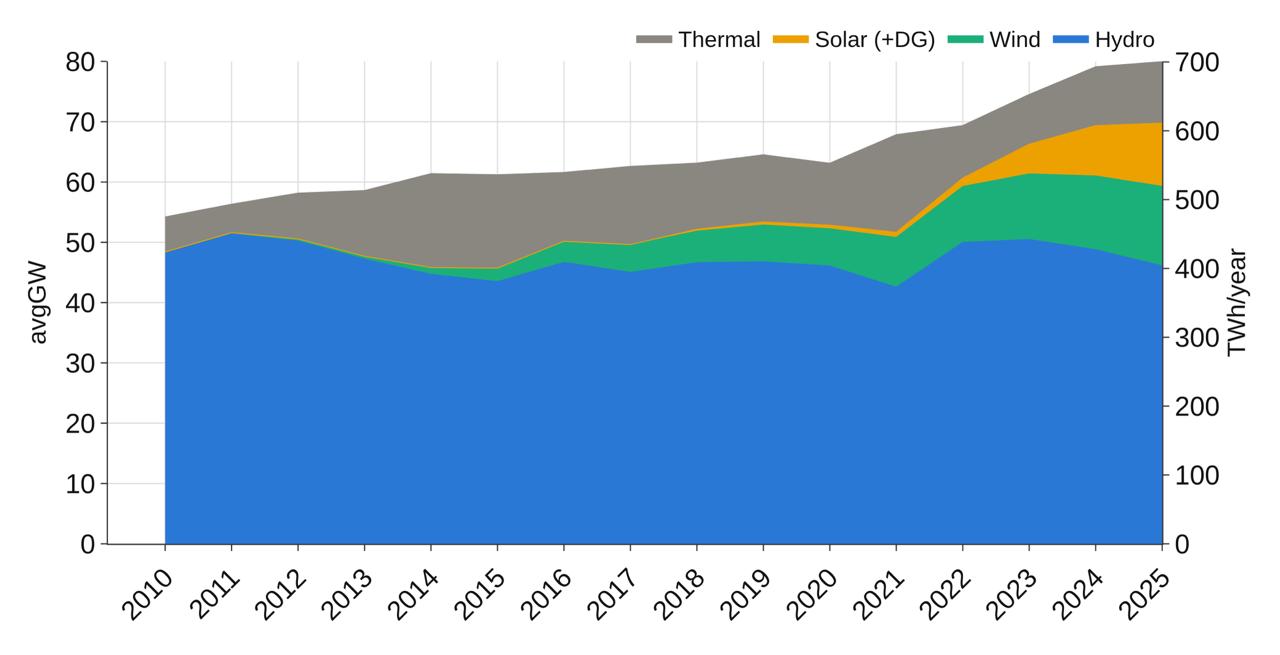}
    \captionof{figure}{Evolution of annual electricity generation in the Brazilian power system by source (average GW and TWh/year).}
    \label{fig:empirical_evidence_generation_avg}
\end{minipage}
\end{figure}

In addition, although battery storage is increasingly discussed in Brazilian planning studies, it is not yet a material centrally dispatched storage resource in the Brazilian power system. Therefore, the country still relies primarily on hydro reservoirs and thermal plants for centrally dispatched controllable energy. As of 2025, the maximum stored-energy capacity of the Brazilian hydro-reservoir system was approximately 292 avgGW, equivalent to more than five times the installed capacity of thermal plants and more than three times the average electricity demand in the country in 2025. Hence, the efficient management of water stored in reservoirs is crucial to ensuring supply adequacy.

In this framework, inflow forecasts become one of the most important inputs to the operation-planning models. However, as discussed in Section~\ref{sec:introduction}, recent evidence shows that the PARp-A model -- the official inflow-forecasting methodology used in Brazil -- has exhibited a systematic optimistic bias. In the terminology of Section~\ref{sec:implications}, positive forecast errors correspond to optimistic shifts in the future inflow process that enter the water balance in the cost-to-go recursion. In Brazilian long-term operation planning, natural inflows are commonly represented in energy units through the Natural Inflow Energy (NIE) metric. The NIE converts the water naturally arriving at the hydroelectric system into the amount of energy that this water could produce along the corresponding hydroelectric cascade (see \cite{brigatto2025assessing} and references therein). In other words, instead of treating inflows only in volumetric units, the planning model expresses them as an energy-equivalent input measured in avgGW. This representation has historically been central to the Brazilian SDDP-based planning framework.

Figure~\ref{fig:empirical_evidence_nie}, an extended version of the figure presented in \cite{brigatto2025assessing} with data updated through October 2026, compares forecast and observed NIE in percentage of the long-term mean for each month for the Northeast subsystem, where the bias pattern is most pronounced. The analogous comparison for the Southeast subsystem is reported in \cite{brigatto2025assessing}. The continuous blue line with circular markers represents the observed values, while each lighter dashed line depicts the first 12-step-ahead point forecasts made from each period in the evaluation horizon.\footnote{Figure~\ref{fig:empirical_evidence_nie} extends the figure of \cite{brigatto2025assessing} to October 2026. Only a brief description is provided here, and readers are referred to the original paper for the complete methodology. Each forecast shown in the rolling evaluation (2011--2026) was produced using the official forecasting model calibrated with all historical observations available at the forecast origin (from January 1931 up to period $t$). Consequently, every forecast curve represents the expected evolution over the subsequent 12 months conditional on the information available at the time the forecast was issued.} A visual inspection already suggests a systematic overestimation of NIE, and therefore an optimistic assessment of future water availability. Figure~\ref{fig:empirical_evidence_bias} provides the sample average and the statistical confidence intervals for the forecast errors for each forecast step ahead ($k$ steps-ahead forecast bias). According to the official data from January 2011 to October 2026, since zero does not belong to the confidence interval, the optimistic bias is statistically significant for both the Southeast and Northeast subsystems. We refer the reader to \cite{brigatto2025assessing} for a detailed discussion and the calculation methodology.

\begin{figure}[!htbp]
    \centering
    \begin{minipage}[t]{0.48\textwidth}
        \centering
        \includegraphics[
            width=\linewidth,
            trim={5.5cm 1.5cm 4.0cm 1.6cm},
            clip
        ]{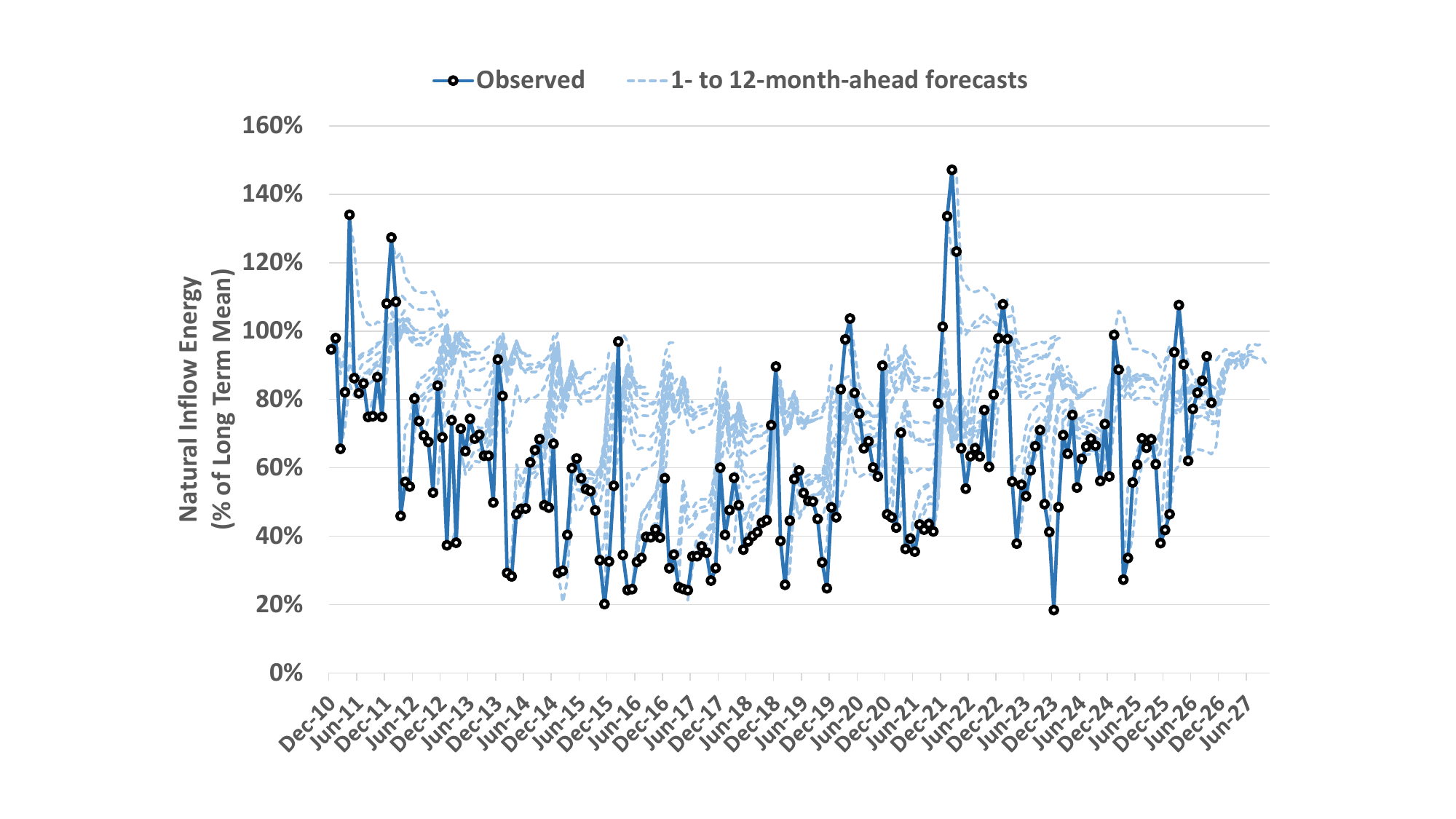}
        \captionof{figure}{PARp-A Forecasts and Observations of NIE for the Northeast subsystem in \% of the seasonal long-term mean.}
        \label{fig:empirical_evidence_nie}
    \end{minipage}
    \hfill
    \begin{minipage}[t]{0.48\textwidth}
        \centering
        \includegraphics[
            width=\linewidth,
            trim={4.7cm 2.0cm 3.2cm 0.5cm},
            clip
        ]{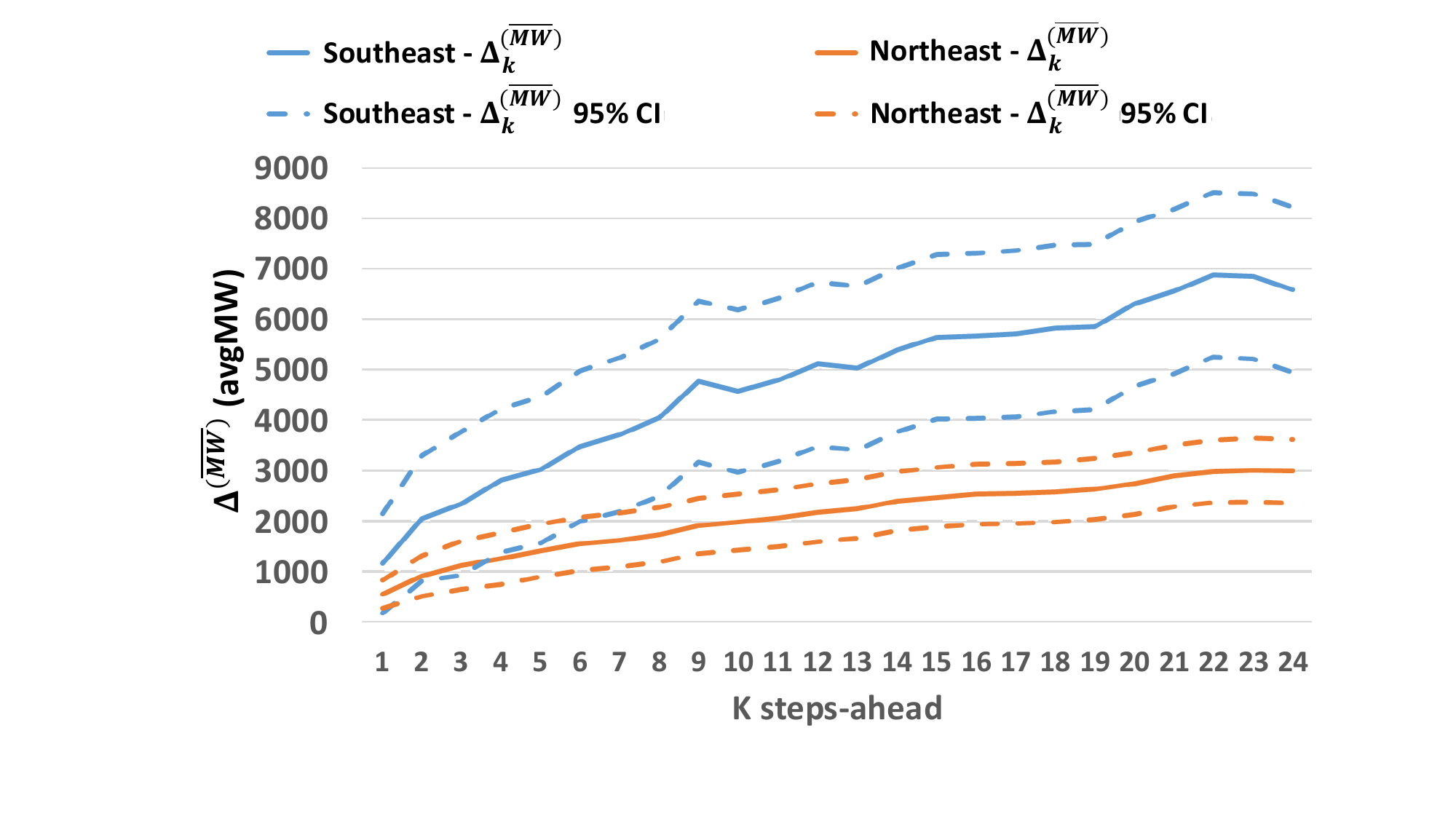}
        \captionof{figure}{Southeast and Northeast NIE Forecasts bias in average MW and 95\%--confidence interval (2.5 and 97.5\% quantiles).}
        \label{fig:empirical_evidence_bias}
    \end{minipage}
\end{figure}

Having established the existence of statistically significant forecast biases (see \cite{brigatto2025assessing}), we now present the original results and analysis of this paper. We begin by comparing planned and implemented trajectories for hydro generation, stored energy, thermal dispatch, and spot prices.

Figure \ref{fig:empirical_evidence_hydro_generation} compares implemented hydropower generation with the 1- to 12-month-ahead average planned values produced by the official SDDP implementation in Brazil. The black line with circle markers denotes implemented values, while the dashed blue lines denote planned values. The bars report 6-month-ahead forecast errors, computed as planned minus implemented generation; red bars therefore indicate periods in which the planning model projected more hydropower generation than was ultimately implemented, while green bars indicate the opposite.

For most of the analyzed period, implemented hydropower generation lies below the planned trajectories. This pattern does not, by itself, allow us to readily attribute the deviations between the two trajectories to forecast bias. Nevertheless, the pattern is consistent with a planning policy that anticipated higher water availability than what ultimately materialized in operation.
Figure~\ref{fig:empirical_evidence_stored_energy} presents the corresponding comparison for stored energy across all reservoirs in the Brazilian power system. The pattern is similar to that observed for hydropower generation: for much of the analyzed period, the planning model projected reservoir trajectories above the levels that were ultimately implemented, which can be related to the abundance of water perceived by the planning model that fails to materialize in practice. Noticeably, the six-month-ahead forecast error reaches values close to 100 avgGW in some periods. This is a very large discrepancy when compared with the Brazilian system's maximum stored-energy capacity of approximately 292 avgGW, since a single six-month-ahead planning error can amount to roughly one third of the total energy that can be stored in the reservoir system. 

\begin{figure}[H]
    \centering

    \begin{minipage}{0.49\textwidth}
        \centering
        \includegraphics[
            width=\linewidth,
            trim={4.1cm 1.5cm 4.0cm 1.6cm},
            clip
        ]{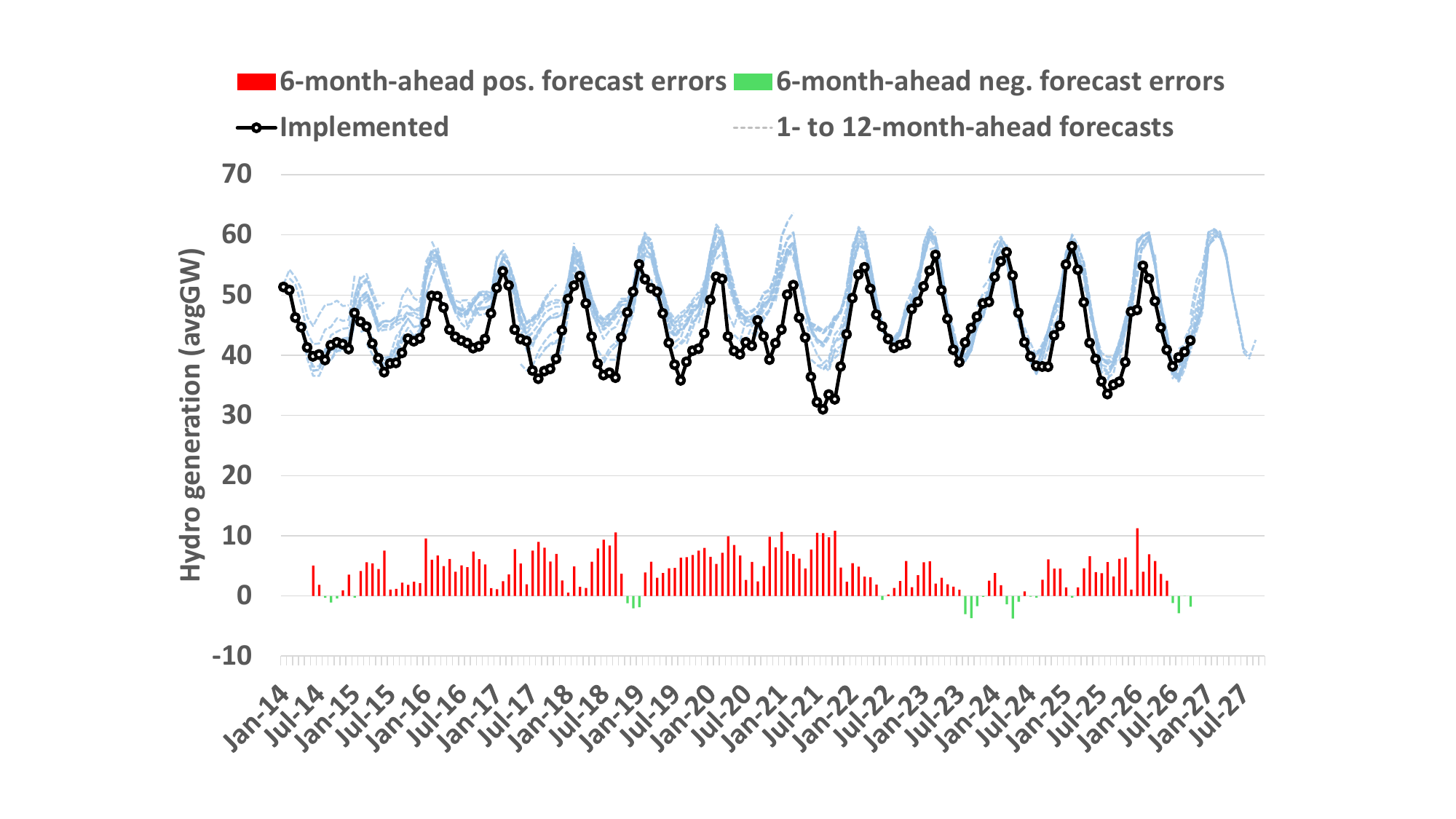}
        \captionof{figure}{Planned and implemented levels of hydro generation for the Brazilian energy system in avgGW.}
        \label{fig:empirical_evidence_hydro_generation}
    \end{minipage}
    \hfill
    \begin{minipage}{0.49\textwidth}
        \centering
        \includegraphics[
            width=\linewidth,
            trim={4.1cm 1.5cm 4.0cm 1.6cm},
            clip
        ]{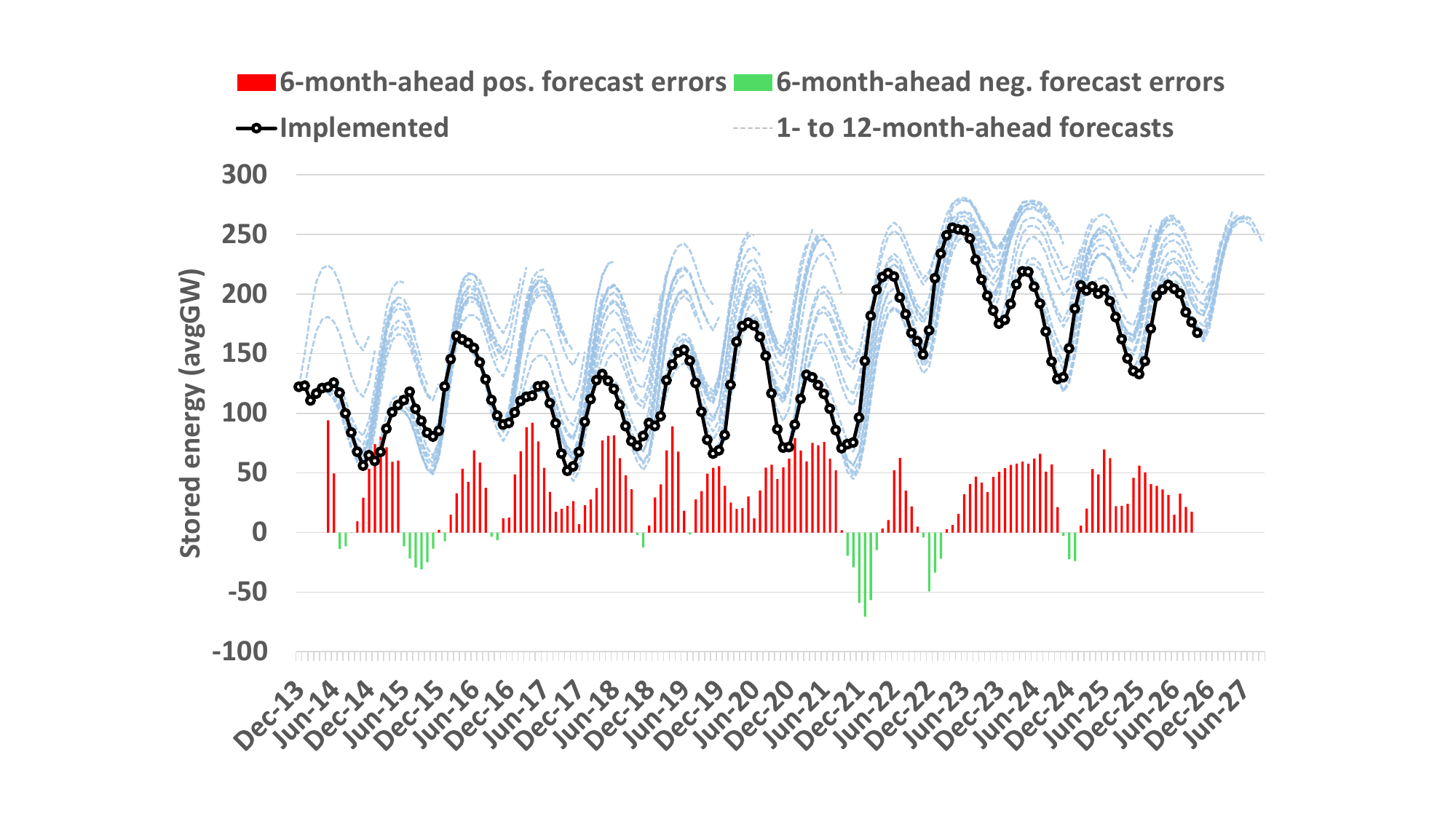}
        \captionof{figure}{Planned and implemented levels of stored energy for the Brazilian energy system in avgGW.}
        \label{fig:empirical_evidence_stored_energy}
    \end{minipage}

\end{figure}

The gap between planned and implemented storage may also reflect operational, regulatory, and modeling factors unrelated to inflow forecasts. To gather statistical evidence on whether inflow forecast errors are nevertheless associated with this gap, we estimate a regression relating cumulative storage planning errors to cumulative NIE forecast errors. Let the $k$-month-ahead NIE forecast error made at time $t$ be defined as $e^{(NIE)}_{t,k} = \hat{y}_{t+k\mid t} - y_{t+k}$, where $\hat{y}_{t+k\mid t}$ is the NIE forecast for period $t+k$ produced with information available at time $t$, and $y_{t+k}$ is the realized NIE. The cumulative $p$-month NIE forecast error is then $E^{(NIE)}_{t,p} = \sum_{k=1}^{p} e^{(NIE)}_{t,k}$. Similarly, we define $E^{(storage)}_{t,p}$ as the cumulative $p$-month-ahead storage planning error, comparing the storage trajectory planned at time $t$ by the operation-planning model with the storage levels subsequently implemented. We then estimate
\begin{align}
    E^{(storage)}_{t,p}
    =
    \beta_0
    +
    \beta_1 E^{(NIE)}_{t,p}
    +
    \Gamma_{m(t)}
    +
    \varepsilon_t,
\end{align}
where $\Gamma_{m(t)}$ denotes monthly dummy variables that model seasonal effects. Table \ref{tab:regression_results} shows overall results for the regression under different values for $p$.

\begin{table}[H]
\centering
\caption{Stored-energy and NIE forecast-error regression results.}
\label{tab:regression_results}
\small
\setlength{\tabcolsep}{6pt}
\begin{tabular}{ccccc}
\hline
$p$ & $\beta_1$ [95\% CI] & HAC $p$-value
& \shortstack[c]{\rule{0pt}{2.6ex}Pearson\\correlation}
& $R^2$ \\
\hline
3  & 1.77 [1.53, 2.00] & $<0.01$ & 0.83 & 0.77 \\
6  & 2.93 [2.49, 3.37] & $<0.01$ & 0.78 & 0.75 \\
12 & 4.57 [3.35, 5.79] & $<0.01$ & 0.76 & 0.64 \\
\hline
\end{tabular}
\end{table}

The results show a strong positive relationship between cumulative NIE forecast errors and cumulative stored-energy forecast errors\footnote{Augmented Dickey--Fuller tests (intercept, no deterministic trend; 3 and 12 lags) reject the unit-root null at the 5\% level for every cumulative forecast-error series, supporting the use of the variables in levels. No statistically significant deterministic linear trend is found at the 3- and 6-month horizons; the 12-month NIE series exhibits a small negative trend that explains only a minor share of its variation. HAC standard errors account for the serial correlation induced by overlapping forecast windows.}. For all horizons, the coefficient on NIE forecast error is positive and statistically significant. This means that a marginal increase in cumulative inflow-forecast bias is statistically associated with a marginal increase in the model's overestimation of stored energy. This marginal effect, $\beta_1$, is greater than one at every horizon. The confidence intervals and HAC p-values indicate that the estimates are statistically significant and robust to serial correlation induced by overlapping forecast windows. The relationship is strongest at shorter horizons and weakens gradually as the forecast horizon increases. This is consistent with the fact that stored energy depends directly on realized inflows, but also reflects operational decisions, system constraints, and hydrological adjustments whose influence may become more relevant over longer forecast horizons. %

The positive and statistically significant coefficients show that optimistic cumulative NIE forecasts are systematically associated with optimistic stored-energy trajectories. This provides empirical evidence consistent with the interpretation that optimistic inflow forecasts affect planned storage trajectories and contribute to the deviations observed in operation.

Taken together, the empirical evidence is consistent with the accumulation of forecast-bias effects in the Brazilian system. Interpreted through the mechanism of Section \ref{sec:implications}, these patterns suggest that the planning policy may undervalue stored water when optimistic inflow expectations fail to materialize. As described by Theorem \ref{theo:bias_increase_discharge}, this undervaluation induces higher current hydropower generation than would be prescribed under unbiased forecasts. Moreover, Figure \ref{fig:empirical_evidence_hydro_generation} shows that the gap between planned and implemented hydropower generation tends to widen during the dry season. In accordance with Remark 2, this pattern indicates that the effect of forecast bias becomes more pronounced under dry hydrological conditions. This effect will be further discussed in Section \ref{sec:case_study}.

The evidence in this subsection is consistent with an affirmative answer to the first question posed in Section \ref{sec:introduction}: optimistic inflow forecasts appear to be associated with lower implied water values and higher planned hydro use than would be expected under realized water availability. The second posed question in Section \ref{sec:introduction} can be partially discussed in this section with thermal generation and electricity spot prices empirical data.

Intuitively, if the planning model relies on more hydropower generation than can be sustained under realized inflows, the opposite gap should appear in thermal dispatch. Figure \ref{fig:empirical_evidence_thermal_generation} is consistent with this pattern. For most of the analyzed period, implemented thermal generation exceeds planned values. The series is also highly volatile, with sharp spikes suggesting that the model delayed thermal dispatch until storage conditions had already deteriorated. This pattern reinforces the interpretation that forecast bias accumulates over time. For instance, a model that anticipates abundant future inflows postpones thermal generation and resorts to it only as a last-minute measure. A similar pattern appears in spot prices, as shown in Figure \ref{fig:empirical_evidence_spot_prices}, extending the interpretation from operational distortions to market distortions. As discussed in Remark 1, the bias is expected to artificially reduce spot prices before scarcity materializes. 

\begin{figure}[!ht]
\centering
\begin{minipage}{0.49\textwidth}
    \centering
    \includegraphics[width=\textwidth, trim={4.1cm 1.5cm 4.0cm 1.6cm},clip]{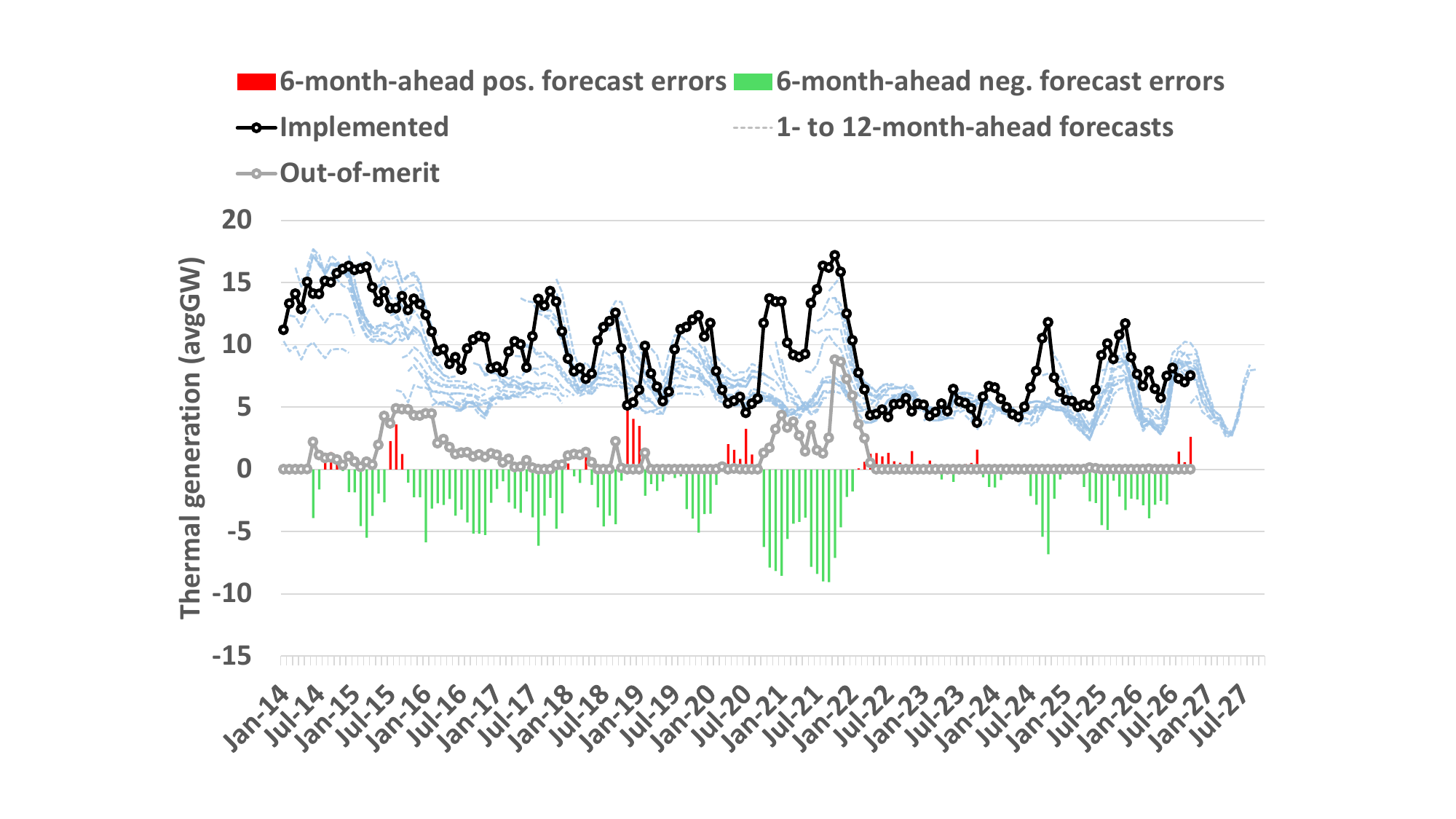}
    \captionof{figure}{Planned and implemented levels of thermal generation for the Brazilian energy system in avgGW.}
    \label{fig:empirical_evidence_thermal_generation}
\end{minipage}
\hfill
\begin{minipage}{0.49\textwidth}
    \centering
    \includegraphics[width=\textwidth, trim={4.1cm 1.5cm 4.0cm 1.6cm},clip]{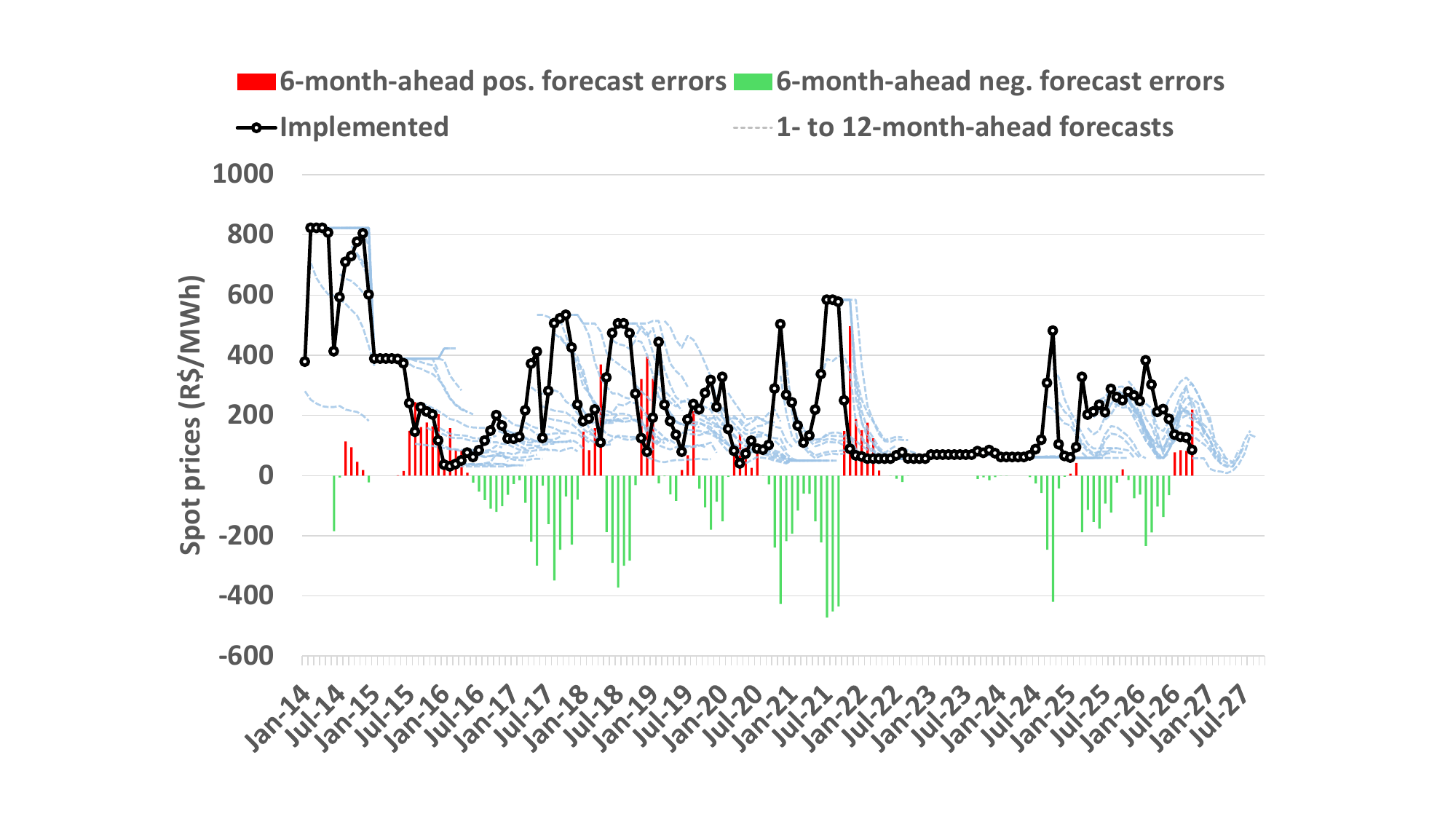}
    \captionof{figure}{Planned and implemented spot prices for the SE subsystem in R\$/MWh.}
    \label{fig:empirical_evidence_spot_prices}
\end{minipage}
\end{figure}

In addition, Figure \ref{fig:empirical_evidence_thermal_generation} also shows the total out-of-merit thermal dispatch for the same period\footnote{The Brazilian ISO classifies out-of-merit thermal dispatch according to different motivations. In Figure \ref{fig:empirical_evidence_thermal_generation}, we consider only out-of-merit dispatch associated with energy supply assurance. The operator explicitly uses this type of dispatch to preserve reservoir storage levels.}. The decision to implement out-of-merit dispatch is made on an ad hoc basis. The operator must justify the need for additional dispatch to the Brazilian Electric Sector Monitoring Committee, which may authorize a specified amount of additional generation for a limited period. Therefore, because there is no explicit rule determining this type of dispatch, we analyze its occurrence through Table \ref{tab:oom}. The table reports, for each quarter, the share of low-storage months in which out-of-merit dispatch was implemented, considering different stored-energy thresholds. The table shows that out-of-merit dispatch is frequently associated with low stored-energy conditions, especially when stored energy falls below 40\% or 50\%. This supports the interpretation that out-of-merit dispatch is often used as an ex-post corrective measure to preserve reservoir levels after the planning policy has relied excessively on hydropower generation.

\begin{table}[H]
\centering
\caption{Frequency of out-of-merit dispatch conditional on low stored energy from January 2014 to October 2026. Each entry reports the share of months with out-of-merit dispatch among the months in which stored energy was below the corresponding threshold. The number of months below the threshold is reported in parentheses.}
\label{tab:oom}
\begin{tabular}{lccc}
\hline
 & \multicolumn{3}{c}{Stored energy threshold} \\
\cline{2-4}
Quarter & $\leq 30\%$ of max capacity & $\leq 40\% $ of max capacity & $\leq 50\%$ of max capacity \\
\hline
Jan--Mar & 83\% $(6)$  & 72\% $(18)$ & 68\% $(25)$ \\
Apr--Jun & --          & 100\% $(5)$ & 88\% $(16)$ \\
Jul--Sep & 100\% $(3)$ & 79\% $(14)$ & 75\% $(20)$ \\
Oct--Dec & 76\% $(17)$ & 79\% $(24)$ & 63\% $(30)$ \\
\hline
\end{tabular}
\end{table}

Figure \ref{fig:empirical_evidence_thermal_generation} indicates that thermal dispatch reached its highest levels in 2021, during the peak of Brazil's most recent severe hydrological crisis. After the end of 2021, thermal generation fell sharply and remained close to the planned trajectories for several months. This coincided with the strong hydrological recovery observed in 2022 and 2023, when higher rainfall and inflows allowed reservoir levels to recover, as shown in Figure \ref{fig:empirical_evidence_stored_energy}. As storage conditions improved, both implemented storage and implemented hydropower generation became more aligned with the planned policies. This behavior is consistent with Remark 4: wet realizations may temporarily mask the effects of forecast bias, making the planning policy appear well calibrated when favorable inflows materialize. The relevant diagnostic pattern, therefore, is not that forecast bias produces visible distortions in every period, but that its accumulated effects become binding when optimistic forecasts are followed by dry realizations.

Taken together, the historical evidence is consistent with the operational and market distortions described in the second research question: lower-than-planned reservoir levels, delayed thermal dispatch, corrective out-of-merit actions, and spot-price deviations when realized inflows fail to match optimistic expectations. These observations alone do not establish the causal impact of the forecast bias, because realized operation also reflects hydrological variability, operational adjustments, regulatory decisions, and model simplifications, as reported in previous works \citep{street2020assessing, rosemberg2021assessing, navarro2024medium}.
To fully assess the effects of the reported inflow forecast bias in Brazil and relate them to the theoretical analysis of Section \ref{sec:implications}, we need controlled experiments that compare the biased planning policy with the counterfactual operation obtained under unbiased inflow forecasts. This is carried out in the next section.

\section{Controlled experiments on the effect of unbiased forecasts} \label{sec:case_study}

Because the empirical evidence in Section~\ref{sec:empirical_evidence} does not provide an observed unbiased counterfactual and the theoretical mechanism of Section~\ref{sec:implications} relies on simplifying assumptions, we now present controlled experiments with realistic system data, designed to isolate the effect of the optimistic inflow-forecast bias identified in \cite{brigatto2025assessing} and assess whether it propagates through the cost-to-go functions computed by the SDDP algorithm, thereby distorting the resulting operating policy. The analysis focuses on reservoir storage, thermal dispatch, total operating costs, and load shedding. We then examine how these operational effects are transmitted to spot prices. Finally, in Section \ref{sec:forward_market}, we analyze how these distortions propagate further to market agents in the forward market. All data, processing scripts, and replication materials used in this Section and in Section \ref{sec:forward_market} are available in \cite{Brigatto2026BiasedInflowsData}.

To identify the operational effects of forecast bias, we need to compare the factual policy actually devised by the biased forecasting methodology with a counterfactual policy that would have been obtained under unbiased inflow expectations. The factual policy is constructed by training the SDDP model with scenarios generated from the official PARp-A methodology, which carries the optimistic bias documented in \cite{brigatto2025assessing}. This policy is hereinafter referred to as Policy 1. The counterfactual policy is constructed using the Adjusted Long-Term Mean (ALTM) PARp-A model proposed in \cite{brigatto2025assessing}, which is shown to produce unbiased forecasts by correcting the long-term mean of the log-transformed inflow process. This policy is hereinafter referred to as Policy 2.

In the SDDP framework, the cost-to-go function of each stage is approximated by the upper envelope of a set of linear cuts, as in the three-segment function of Figure~\ref{fig:shift_illustration}. The negative of the slope of each cut with respect to storage is a water value. Policies 1 and 2 differ only in the inflow scenarios used to compute these cuts. Under the optimistic scenarios used to train Policy 1, the cuts are computed as if more water were available in the future stages, which is the mechanism illustrated in panel (b) of Figure~\ref{fig:shift_illustration} for the stylized case. The operational decisions of each stage are then obtained by solving the stage problem of \ref{sec:appendix_sddp} with these cuts. Therefore, lower water values translate into higher hydro generation at a given storage level and into lower spot prices before scarcity materializes, as examined in Sections~\ref{sec:operation_results} and~\ref{sec:spot_price_results}, respectively.

Both policies are then evaluated out of sample under the same set of realistic hydrological scenarios generated with the ALTM model. This design holds the realized inflow conditions fixed across policies, so that differences in system operation, spot prices, operating costs, and load shedding can be attributed to the stochastic process used to train the SDDP policy, rather than to differences in ex post hydrology or system operation details.

The experiments are based on real data from the official operational planning models used in Brazil, covering monthly data from January 2021 to December 2025 (with an additional 60 months included to reduce end-horizon effects). The dataset includes the complete set of hydroelectric plants, represented individually, along with their cascade interconnections, storage capacities, energy conversion factors, and operational limits, such as maximum and minimum outflows and turbine capacities. Operational information is also included for all thermal power plants, energy exchange, and demand. The system representation follows the classical four-subsystem division adopted in Brazilian operation planning. The modeling follows the classical hydrothermal planning SDDP formulation presented in \ref{sec:appendix_sddp}.

The dataset described above gives rise to a large-scale stochastic optimization problem, requiring careful tuning of the SDDP parameters to ensure convergence and numerical stability. Several sampling schemes have been proposed in the literature (see \cite{de2015improving,fullner2025stochastic}) approaches, which have been shown to accelerate convergence in certain applications. Nevertheless, to maintain methodological consistency with the official implementation of the SDDP algorithm used in Brazil, we adopt a fixed sampling scheme with 200 forward and 30 backward scenarios.

Several stopping criteria have been proposed for SDDP. The original rule of \cite{pereira1991multi}, based on the lower bound entering the confidence interval of the upper-bound estimator, is the criterion historically adopted in Brazilian hydrothermal operation planning. However, this test is known to be imperfect and may stop prematurely, as discussed by \cite{shapiro2011analysis}. More generally, there is no universally accepted finite-sample convergence test that certifies exact convergence of large-scale SDDP policies in applications of this size \cite{homem2011sampling, fullner2025stochastic}. In Brazilian practice, the original convergence test is still officially used, but it is combined with operational safeguards, including minimum and maximum numbers of iterations. In recent years, in the official operational setting, the algorithm has typically been run with a minimum of 30 and a maximum of 50 iterations. In our experiments, we fix the number of iterations at 50 for both policies. This choice places both policies at the upper end of the official iteration window, avoids differences in stopping behavior across the biased and bias-corrected models, and ensures that the comparison is based on a common and sufficiently long training process. Our objective is not to claim exact convergence of either policy, but to compare two policies trained under the same computational protocol and evaluated under the same out-of-sample hydrological scenarios.

Three elements support the reliability of the numerical comparison. First, the two policies are trained using the same sampling scheme and the same number of iterations, namely, 50 (the maximum number used in official cases), so that any residual optimality gap affects both policies. Moreover, we obtained a stable lower bound, with variations within 0.1\%, for all cases. Second, we evaluate the two policies under the same 2,000 out-of-sample scenarios and compute the statistics of the differences using paired observations. As reported in Table~\ref{tab:ControlledExperiment_costs}, the 99\% confidence interval for the cost difference has a half-width below 3\% of its sample average. Third, the lower limit of the confidence interval for the cost difference between the two policies is much larger than the observed variation of the lower bound. The difference is therefore not only statistically significant but also far larger than the variation of the lower bound of the trained policies, so that the conclusions are unaffected by it. In addition, the convergence of each policy was inspected not only in terms of the objective-function bounds, but also in terms of primal and dual stability. In the last two iterations, the monthly profiles of reservoir storage and spot prices were graphically indistinguishable. Finally, \ref{sec:appendix_sensitivity} repeats the controlled experiments for different initial storage levels and thermal capacities.

Throughout the analysis of the results, we identify the presence of the theoretical mechanisms developed in Section~\ref{sec:implications} in the large-scale system. Section~\ref{sec:operation_results} examines the conditional discharge ordering of Theorem~\ref{theo:bias_increase_discharge}, its heterogeneity across hydrological conditions (Remark 2), and the dynamic effects on storage, thermal dispatch, and operating costs (Remarks 3 and 4). Section~\ref{sec:spot_price_results} examines the effect on spot prices described in Remark 1.

\subsection{Operation results} \label{sec:operation_results}

Theorem \ref{theo:bias_increase_discharge} states that, in the stylized setting of Section~\ref{sec:implications}, the conditional hydro discharge of the policy whose water values are distorted by the optimistic inflow bias---Policy 1 in our setting---must be weakly larger than that of the unbiased counterpart, Policy 2. Because this case study considers a more realistic description of the system, it departs from the stylized single-reservoir setting of Section~\ref{sec:implications} in several material ways: a transmission network, cascades of individually represented reservoirs, multiple hydro plants with heterogeneous productivities, and operational constraints that may, in principle, dampen the bias channel before it reaches the first-stage dispatch decision. Despite these complications, the non-parametric regressions\footnote{We discretize the initial reservoir storage $v_{t-1}$ along a grid; at each grid point, the plotted value is the sample average of all simulated dispatch decisions whose initial storage lies within a window of $\pm 7.5$~avgGW around that point. See \cite{gyorfi2002distribution} for the underlying theory of non-parametric regression.} reported in Figure~\ref{fig:ControlledExperiment_hydro} make clear that the mechanism identified in Theorem~\ref{theo:bias_increase_discharge}  is predominant in practice. For each month of the third year of the planning horizon, conditional on the initial reservoir storage (x-axis), Policy 1 turbines more than Policy 2.
\begin{figure}[h]
    \includegraphics[scale = 0.4, trim={0.1cm 0.0cm 0.0cm 0.0cm},clip]{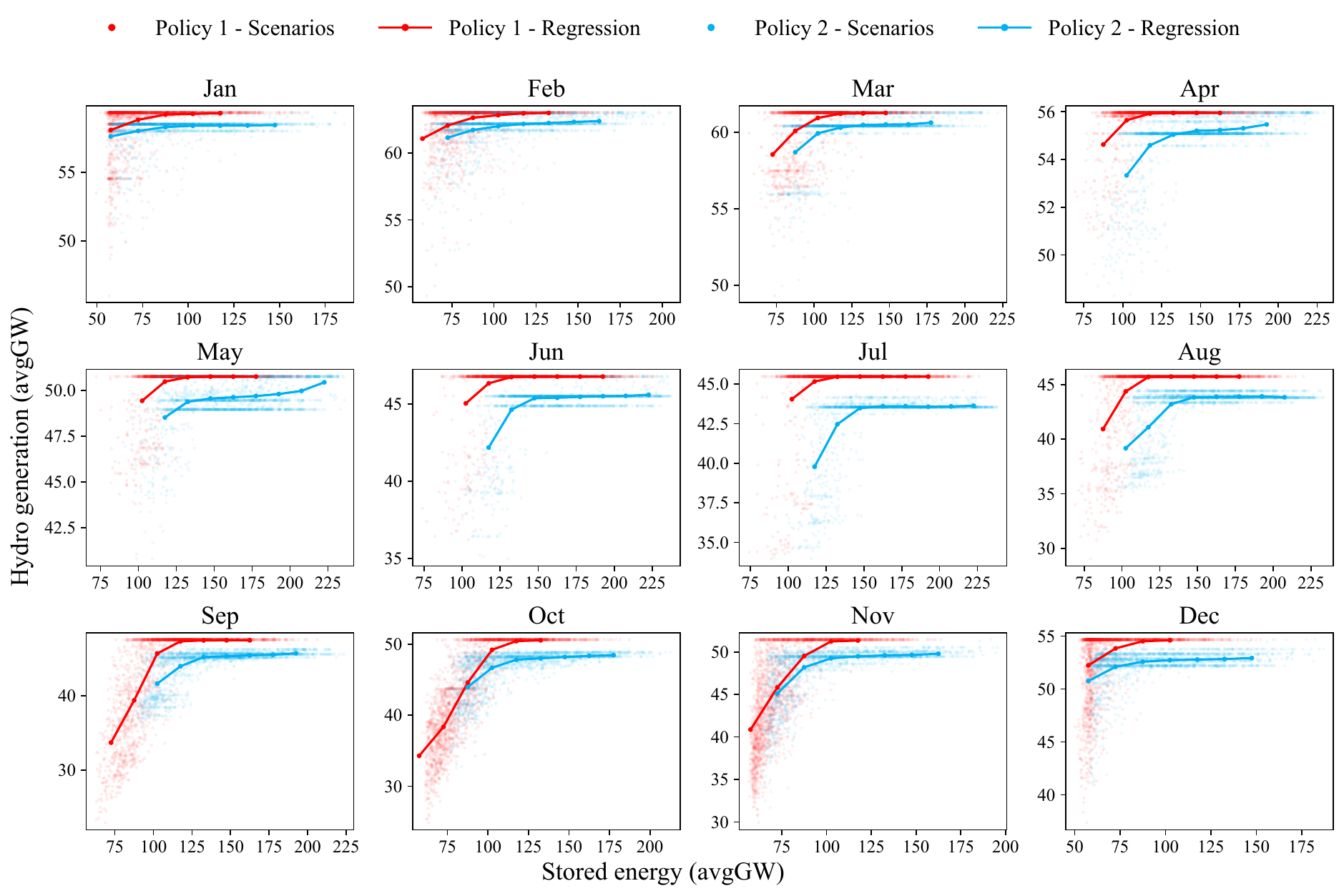}
    \centering
    \caption{Non-parametric regressions of hydro generation on initial reservoir storage, by month of the third year of the planning horizon, under Policies 1 and 2.}
    \label{fig:ControlledExperiment_hydro}
\end{figure}

The gap is small during the wet season, when the water value is low under both policies---as reflected in the wet-season spot prices of Figure~\ref{fig:ControlledExperiment_spot}---and widens sharply through the dry season, when storage becomes scarce and the water-value gap is largest. This wet/dry heterogeneity is exactly what Remark 2 anticipates: in wet periods, the cost-to-go function is approximately flat in $v_{t-1}$, so a bias in the inflow forecast moves the slope of $Q_t$ only marginally. In dry periods, $Q_t$ has a steep negative slope, and the bias produces a substantial reduction in this slope, translating into a much larger conditional over-discharge. So, the effect is larger when it is most critical to the system. The comparison itself (between policies evaluated at similar realized state $v_{t-1}$) is precisely the conditional effect that Theorem \ref{theo:bias_increase_discharge} prescribes for the rolling-horizon implementation. 

As discussed in Remark~3, this is only valid for the conditional comparison. In Figures~\ref{fig:ControlledExperiment_hydro_mean} and \ref{fig:ControlledExperiment_hydro_difference}, however, we also observe what Remark~3 describes as the dynamic effect: in some periods, hydro generation can be, on average, higher under Policy~2 because the biased Policy~1, operating from the lower reservoir states it induced (see Figure~\ref{fig:ControlledExperiment_storage}), faces higher water opportunity costs and therefore reduces the hydro generation. This is a salient feature of Policy~2, which, informed by a more realistic look-ahead and unbiased water values, shifts more water from the wet to the dry season to prevent high-cost thermal dispatches.

\begin{figure}[h]
    \centering
    \begin{minipage}{0.49\textwidth}
        \centering
        \includegraphics[
            width=\textwidth,
            trim={4.1cm 1.6cm 4.0cm 1.6cm},
            clip
        ]{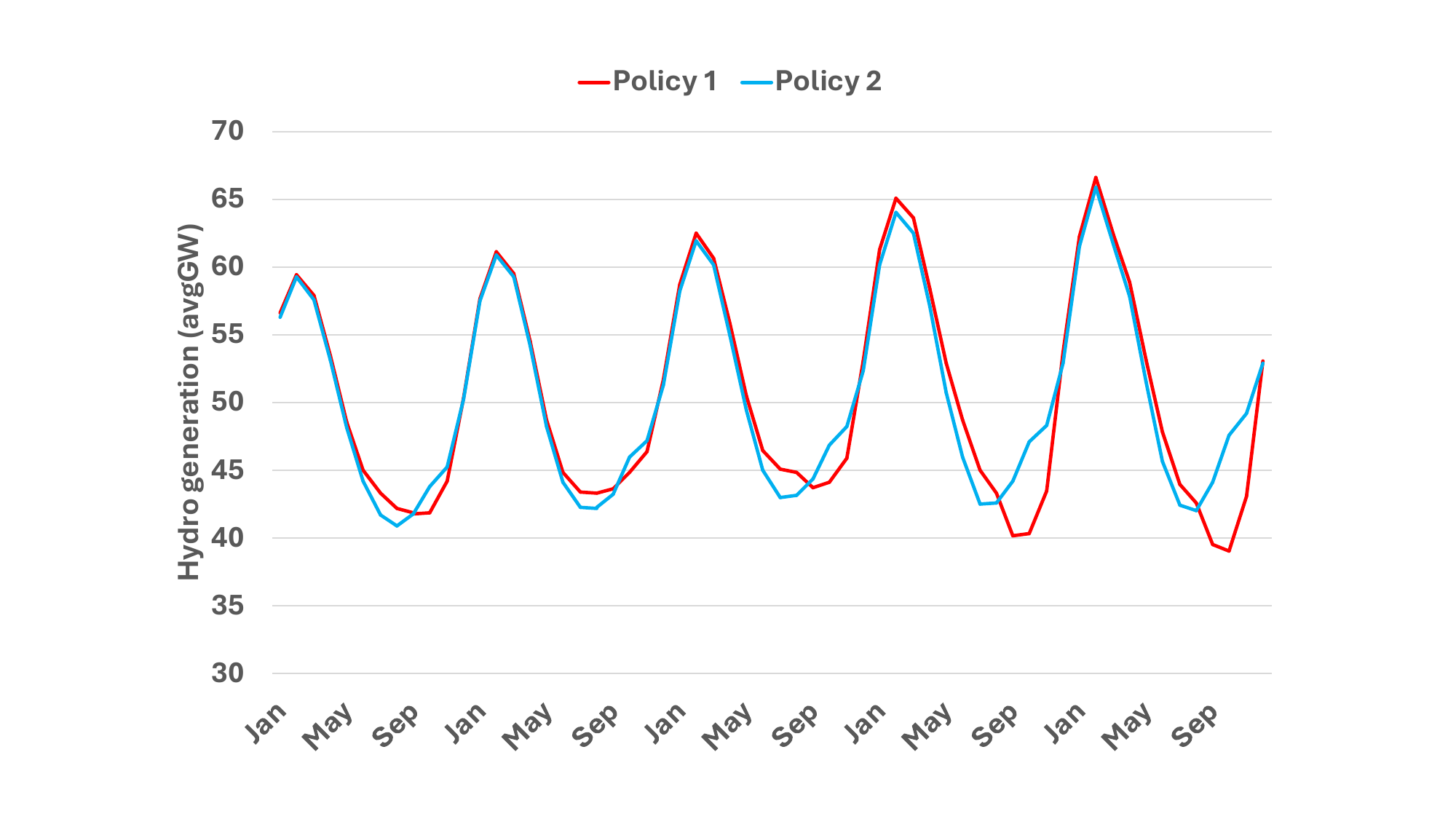}
        \captionof{figure}{Sample-average hydro generation under Policies 1 and 2 across the 2,000 out-of-sample simulations.}
        \label{fig:ControlledExperiment_hydro_mean}
    \end{minipage}
    \hfill
    \begin{minipage}{0.49\textwidth}
        \centering
        \includegraphics[
            width=\textwidth,
            trim={4.1cm 1.6cm 4.0cm 1.6cm},
            clip
        ]{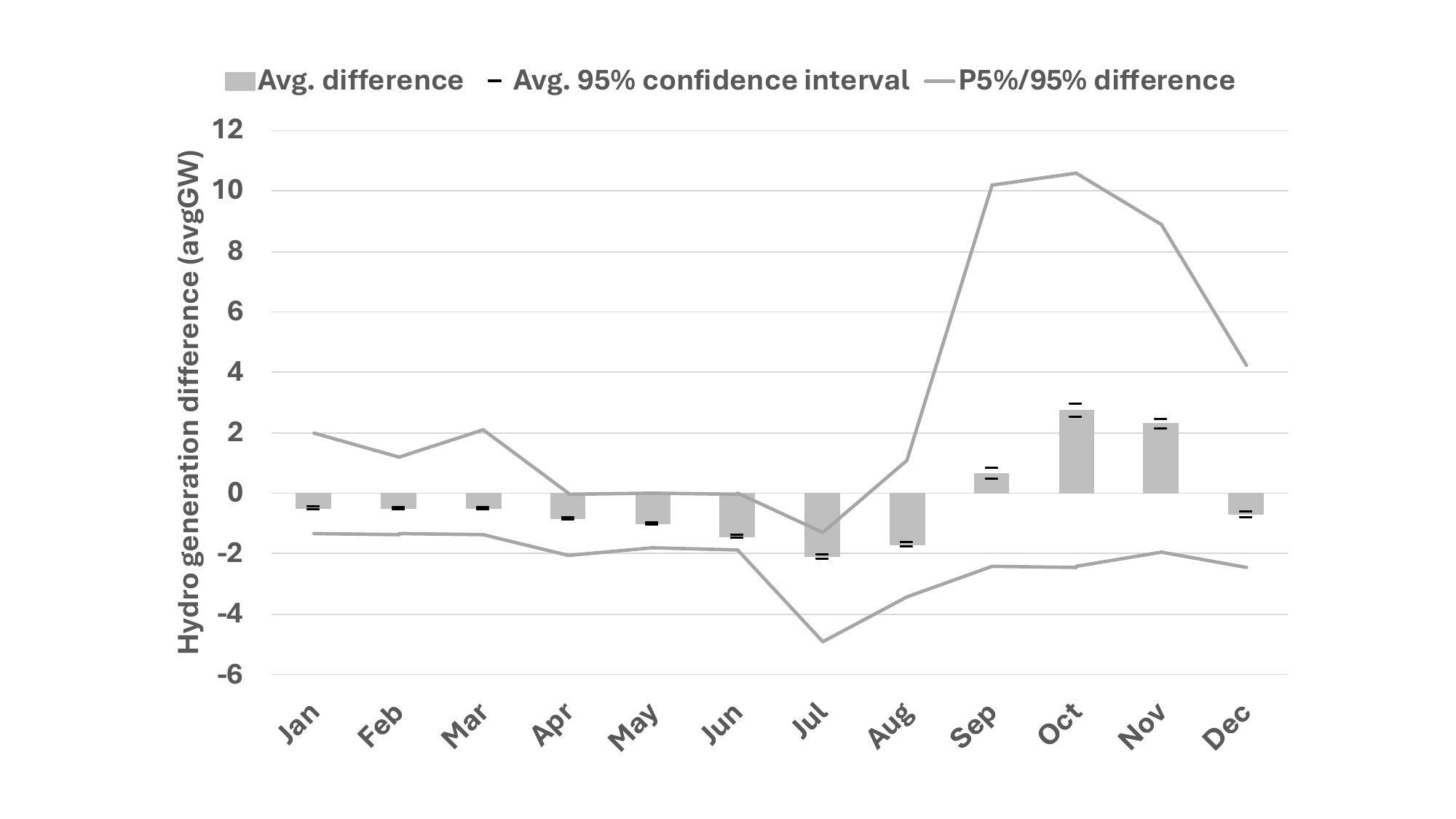}
        \captionof{figure}{Difference in hydro generation in the third year, computed as Policy 2 minus Policy 1, across the 2,000 out-of-sample simulations. All statistics are evaluated using the paired differences between policies for each period and scenario.}
        \label{fig:ControlledExperiment_hydro_difference}
    \end{minipage}
\end{figure}

While the conditional gap in hydro generation per individual month is modest, the resulting sample-average reservoir trajectories shown in Figure~\ref{fig:ControlledExperiment_storage} reveal that the cumulative effect is large. Policy~1 yields lower stored-energy levels in all periods, which is the cumulative effect anticipated in Section~\ref{sec:implications} as a consequence of Theorem~\ref{theo:bias_increase_discharge} under a rolling-horizon implementation. The difference between the two sample averages reaches 8.12 avgGW in September of the first year and 29.09 avgGW in August of the third year, corresponding to about 35\% of the average system demand in 2025. This persistent reduction in stored energy exposes the system to greater operational risk, since it leaves fewer hydro resources available to respond to adverse inflow realizations during critical dry periods.

The thermal generation profile in Figure~\ref{fig:ControlledExperiment_thermal} is consistent with the hydro-generation and storage results. As a preventive response to the possibility of severe scarcity later in the dry season, Policy 2 keeps average thermal dispatch slightly higher than Policy 1 throughout the wet season, which is consistent with the first-stage thermal dispatch ordering that follows from Theorem~\ref{theo:bias_increase_discharge}. After each May, at the end of the wet season, Policy 2 increases thermal generation earlier than Policy 1, complementing the sharper reduction observed in hydropower generation. This anticipatory thermal dispatch preserves stored energy before the peak of the dry season, allowing the system to rely less on stronger corrective thermal dispatch later and to increase hydropower generation again while still in the dry season. Policy 1, by contrast, delays this thermal adjustment until scarcity risk becomes stronger and largely unavoidable, which is consistent with the dynamic feedback described in Remark 3. As a result, it must sharply increase thermal dispatch during the dry season to levels above those observed under Policy 2. 

\begin{figure}[!htbp]
    \centering

    \begin{minipage}{0.49\textwidth}
        \centering
        \includegraphics[
            width=\linewidth,
            trim={4.1cm 1.6cm 4.0cm 1.6cm},
            clip
        ]{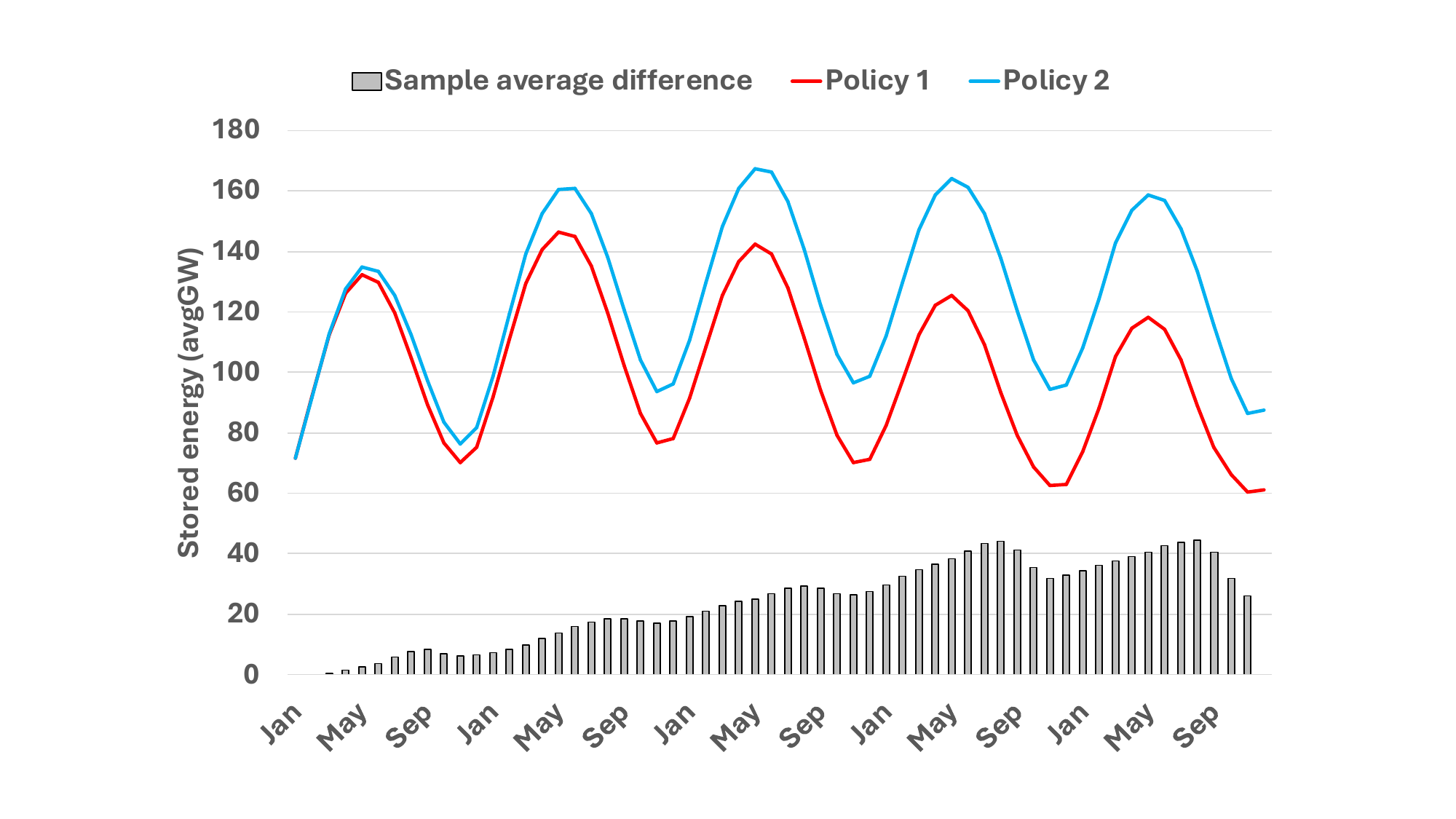}
        \captionof{figure}{Sample-average stored energy under Policies 1 and 2 across the 2,000 out-of-sample simulations.}
        \label{fig:ControlledExperiment_storage}
    \end{minipage}
    \hfill
    \begin{minipage}{0.49\textwidth}
        \centering
        \includegraphics[
            width=\linewidth,
            trim={4.1cm 1.6cm 4.0cm 1.6cm},
            clip
        ]{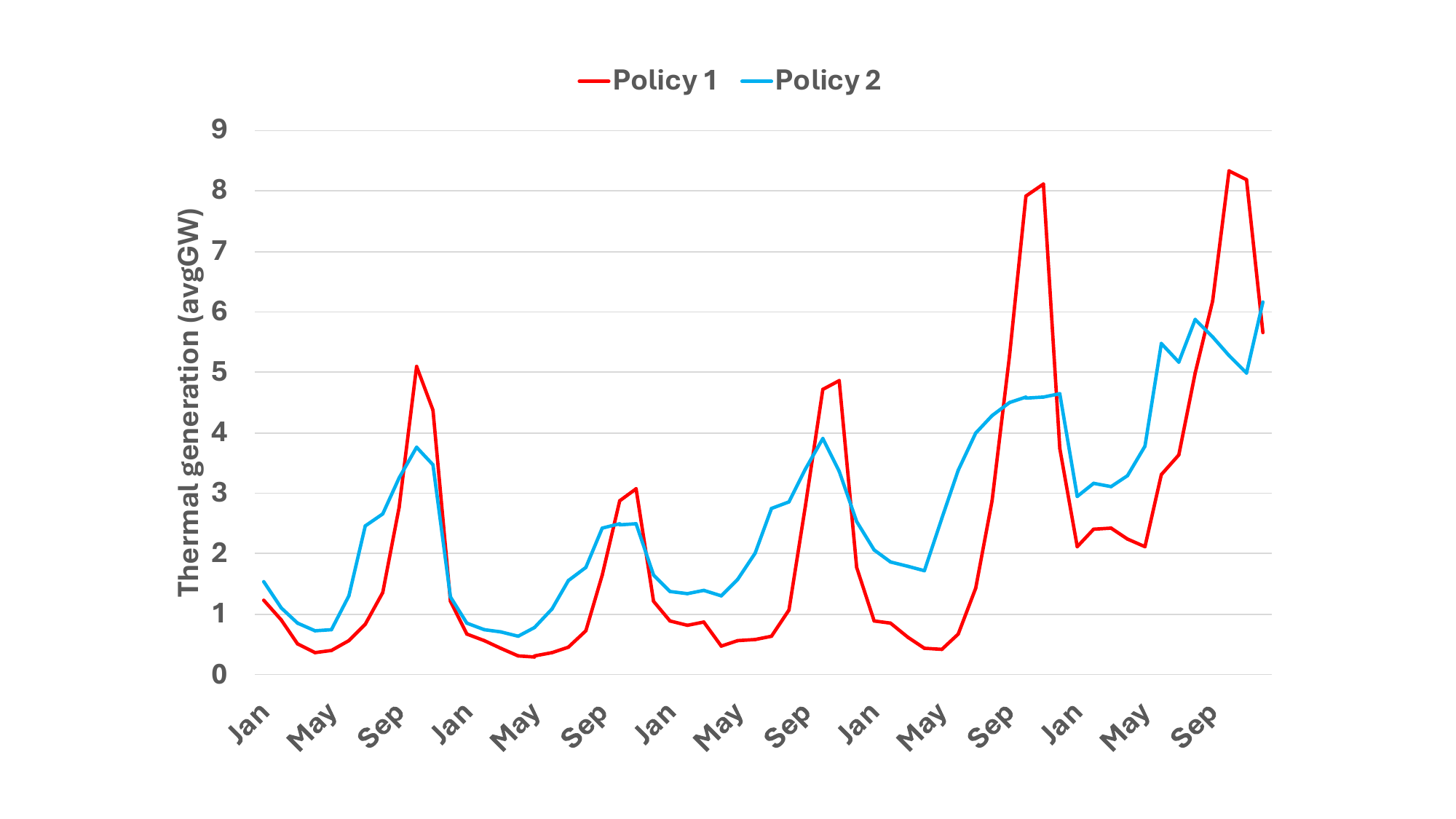}
        \captionof{figure}{Sample-average dispatchable thermal generation under Policies 1 and 2 across the 2,000 out-of-sample simulations.}
        \label{fig:ControlledExperiment_thermal}
    \end{minipage}

\end{figure}

Notably, during the dry-season thermal-generation peak, Policy 1 relies, on average, on more expensive and more CO2-emitting thermal units (oil and diesel) than Policy 2 does throughout most of the years. This follows directly from the dispatch merit order: once cheaper thermal units are already dispatched, additional thermal generation must come from higher-cost units. Therefore, the delayed adjustment under Policy 1 increases operation costs not only because more thermal generation is required in critical periods, but also because this generation is supplied by costlier plants.

The cost impact is further amplified by load shedding. Since thermal capacity is finite, sufficiently severe reservoir depletion cannot be offset indefinitely by additional thermal dispatch. Beyond a certain point, the system must therefore resort to load shedding, whose penalty cost is much higher than regular thermal generation costs. Table \ref{tab:ControlledExperiment_costs} reports sample statistics for cumulative five-year operation costs under both policies, including thermal generation and load-shedding costs. The sample-average cumulative operation cost under Policy 1 is about seven times larger than under Policy 2. This result is consistent with Remark 3, according to which the biased trajectory incurs a higher integrated thermal cost and faces an increased risk of storage exhaustion on dry horizons. The cost difference is computed as the cost under Policy 1 minus the cost under Policy 2, and as reported in Table \ref{tab:ControlledExperiment_costs}, the difference between the two sample averages is statistically significant.

Economically, this result shows that optimistic forecasts do not merely alter dispatch quantities; they transfer costs over time, replacing relatively inexpensive preventive actions with expensive corrective actions and reliability failures later in the horizon. These higher operation costs and reliability risks must ultimately be reflected in electricity prices, which is the focus of the next subsection.

\begin{table}[!htbp]
\centering
\caption{Sample statistics for cumulative five-year operation costs, cost differences, and load-shedding occurrence under Policies 1 and 2 across the 2,000 out-of-sample simulations.}
\resizebox{\textwidth}{!}{%
\begin{tabular}{cccc|cc|c}
         & \multicolumn{3}{c|}{\begin{tabular}[c]{@{}c@{}}Five-year total thermal generation\\ and deficit costs (MMR\$)\end{tabular}} & \multicolumn{2}{c|}{\begin{tabular}[c]{@{}c@{}}Five-year total thermal generation\\ and deficit cost difference (MMR\$)\end{tabular}} & \multirow{2}{*}{\begin{tabular}[c]{@{}c@{}}Share of scenarios with \\ load shedding\end{tabular}} \\ \cline{2-6}
         & P5\%                                  & Sample Average                             & P95\%                                  & Sample average                                                     & 99\%-confidence interval error                                            &                                                                                                   \\ \hline
Policy 1 & 43,683.89                             & 208,239.53                                 & 392,959.51                             & \multirow{2}{*}{179,450.41}                                        & \multirow{2}{*}{$\pm$ 5,050.81}                                        & 10.3\%                                                                                              \\
Policy 2 & 7,704.94                              & 28,789.12                                  & 90,392.76                              &                                                                    &                                                                  & 0.9\%                              \end{tabular}%
}
\label{tab:ControlledExperiment_costs}
\end{table}
\subsection{Spot price results} \label{sec:spot_price_results}
Figure \ref{fig:ControlledExperiment_spot} shows the monthly and yearly sample averages of spot prices under both policies.\footnote{As discussed in Section \ref{sec:implications}, spot prices in Brazil are obtained from the dual variable of the energy balance constraint. In practice, a cap and a floor are applied to the resulting dual variable. In this work, we consider the official values of R\$57.31/MWh and R\$1611.04/MWh adopted in the Brazilian market in 2026.} The pattern is similar to that observed for thermal generation. Policy 2 keeps prices higher than Policy 1 for most of the year, with prices gradually increasing after May. Policy 1, on the other hand, keeps prices at lower levels for most of the year, but exhibits sharp peaks during the dry season. These peaks reflect the delayed dispatch of expensive thermal units and the occurrence of load-shedding costs. Interestingly, the two policies present similar yearly average values for the first two years of operation, but the price under Policy 1 starts to rise from the third year onward, as the reservoir depletion induced by Policy 1 accumulates (Remark 3). Overall, the average spot prices over the five years are 311.47 R\$/MWh under Policy 1 and 232.18 R\$/MWh under Policy 2. This indicates that the two policies differ in both the long-run average price level and the price-spike distributions.

This pattern is consistent with Remark 1. The optimistic forecasts used to construct Policy 1 initially depress water values and, therefore, spot prices, especially at the beginning of the dry season. Note the pattern of lower prices under Policy~1 from May to August compared with Policy~2. However, this lower price signal is artificial: it reflects an overvaluation of conditional future water availability rather than an actual improvement in system conditions. As the resulting reservoir depletion accumulates, the system later faces higher scarcity costs, which appear as sharp dry-season price peaks.

\begin{figure}[!htbp]
    \centering

    \begin{minipage}{0.49\textwidth}
        \centering
        \includegraphics[
            width=\linewidth,
            trim={4.1cm 1.5cm 4.0cm 1.6cm},
            clip
        ]{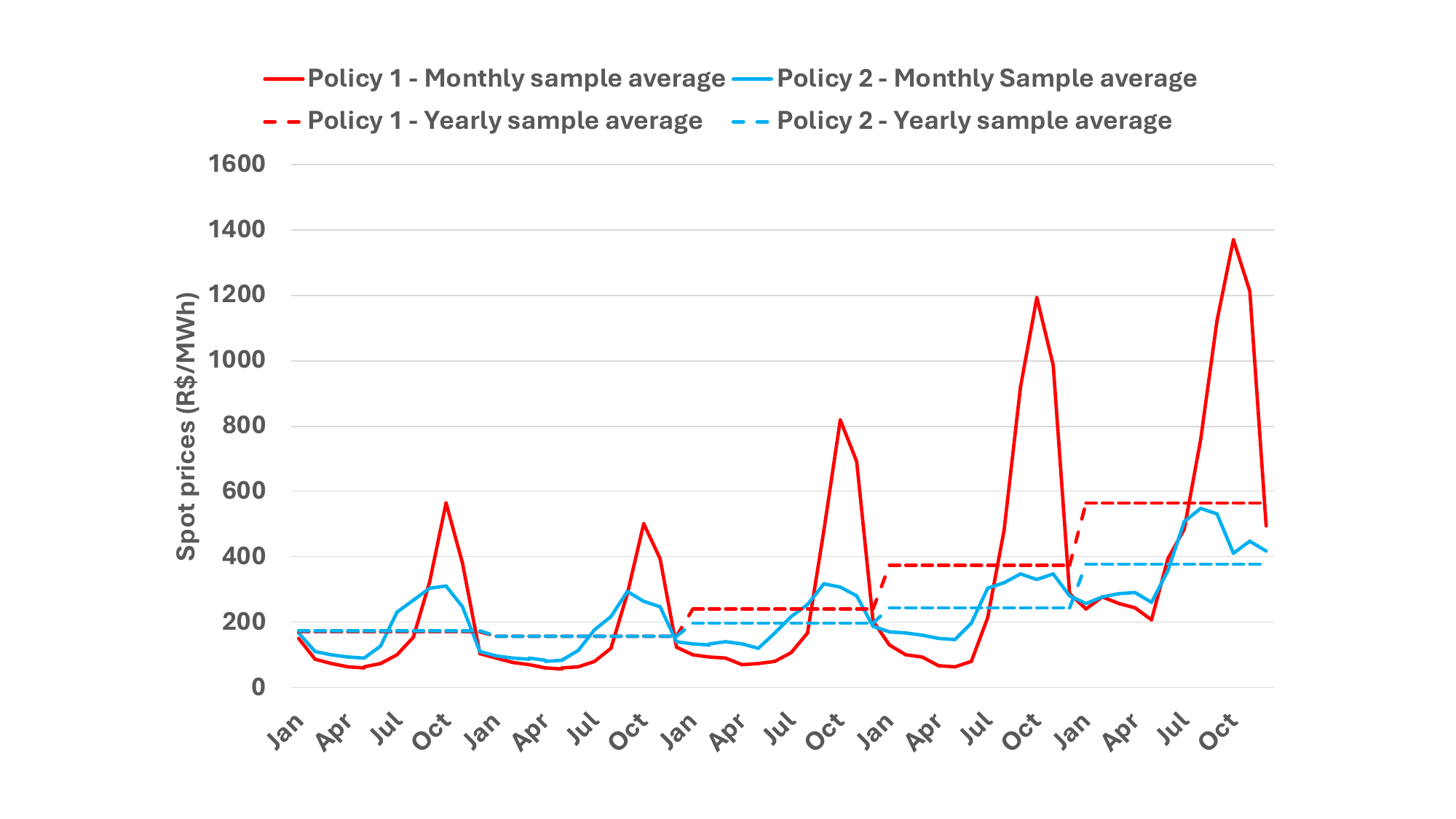}
        \captionof{figure}{Sample-average spot prices under Policies 1 and 2 across the 2,000 out-of-sample simulations. Official cap and floor for the year of 2026 are adopted.}
        \label{fig:ControlledExperiment_spot}
    \end{minipage}
    \hfill
    \begin{minipage}{0.49\textwidth}
        \centering
        \includegraphics[
            width=\linewidth,
            trim={4.1cm 1.6cm 4.0cm 1.6cm},
            clip
        ]{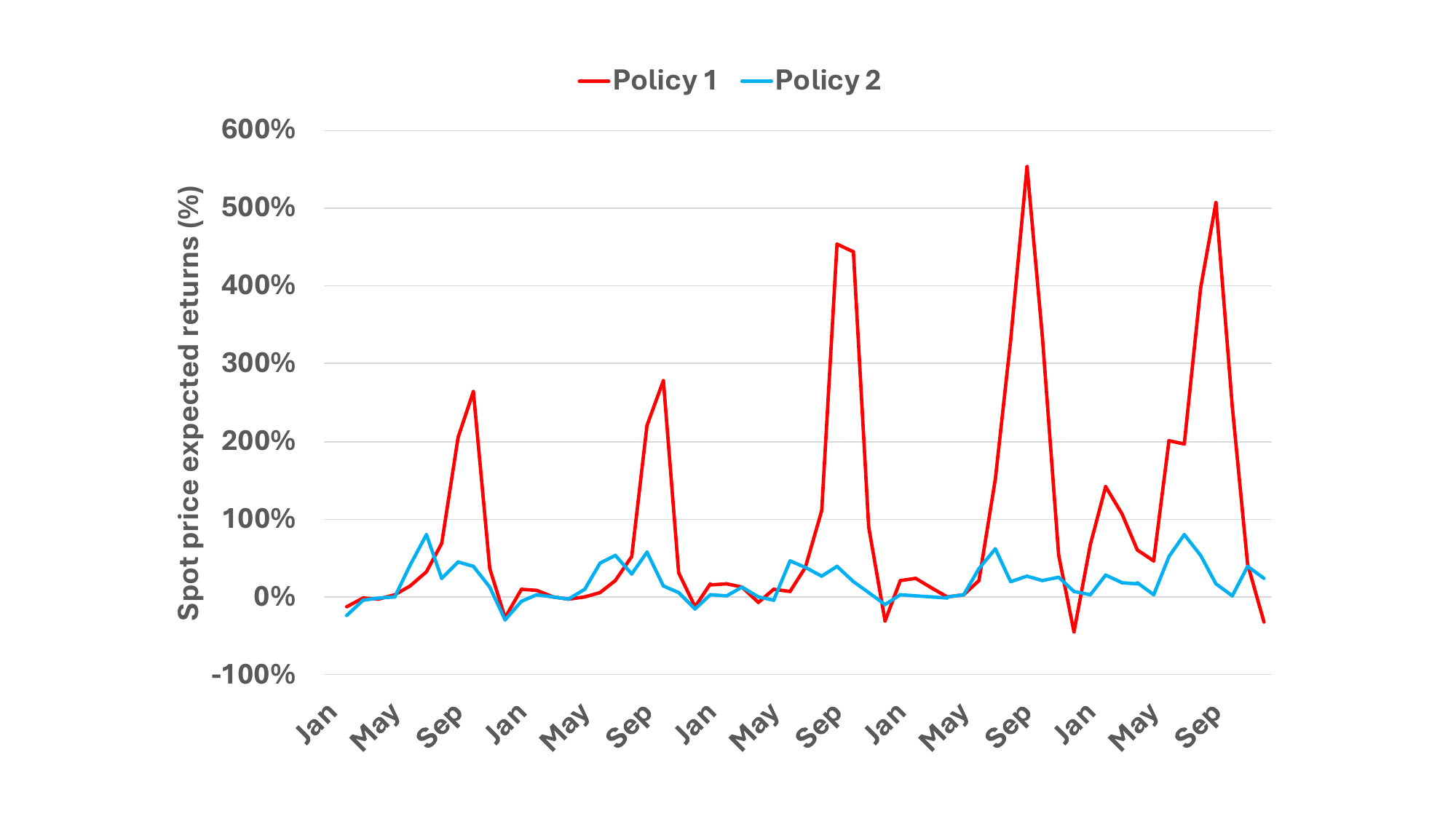}
        \captionof{figure}{Spot-price expected returns (sample averages) under Policies 1 and 2 across the 2,000 out-of-sample simulations.}
        \label{fig:ControlledExperiment_spot_vol}
    \end{minipage}

\end{figure}

We evaluate the spot-price return between consecutive months as $\Delta\pi_t/\pi_{t-1}$, where $\Delta\pi_t = \pi_t - \pi_{t-1}$\footnote{The objective of this metric is to measure the variability of spot prices over time.}, a relevant market index that measures the risk premium agents will be asking for. Figure \ref{fig:ControlledExperiment_spot_vol} presents the spot-price expected returns, given by the monthly sample-average values of $\Delta\pi_t/\pi_{t-1}$, under both policies. The expected returns reproduce the seasonal pattern observed in the spot-price results. They tend to increase as the wet season ends and the system approaches the dry months, and then decrease again as the next wet season begins. In both policies, expected returns are higher during the dry season than during the wet season, reflecting the increasing scarcity of water resources and the possibility of dispatching more expensive thermal units. However, the magnitude of these variations is substantially larger under Policy 1 than under Policy 2 during the dry season. This reinforces that Policy 1 reacts to scarcity later, leading to abrupt price adjustments when the system is already under stress. This delayed response is consistent with the sharper increase in expensive thermal generation and, in some scenarios, with the occurrence of load shedding as discussed in the previous subsection. By contrast, as Policy 2 anticipates the deterioration in system conditions, it preserves more stored water and produces a smoother spot-price trajectory, with smaller month-to-month price changes.

\section{Propagation to the forward-market layer} \label{sec:forward_market}

The controlled experiments of Section~\ref{sec:case_study} isolate the operational consequences of forecast bias on storage, dispatch, and spot prices. Because the same spot prices that emerge from centralized dispatch are also used for short-term settlement in the Brazilian contracting environment, the bias may propagate one step further into the forward-market layer. We examine this channel through a stylized contracting overlay applied to the same SDDP simulations of Policies 1 and 2, asking how forecast-induced changes in price–quantity risk affect the downside-risk-adjusted revenue of a representative hydropower producer at different forward-contracting levels.

We construct a normalized representative hydropower generation profile as follows. Let $G_{t,\omega}^{(k)}$ denote total hydropower generation in period $t$ under inflow scenario $\omega$ and Policy~$k$, and let $\overline{G}^{(k)}$ denote its sample mean across all periods and scenarios. The per unit representative plant's generation is defined as $g_{t,\omega}^{(k)} = G_{t,\omega}^{(k)}/\overline{G}^{(k)}$. Depending on the policy environment, the profile is obtained either from Policy~1, $g_{t,\omega}^{(1)}$, or from Policy~2, $g_{t,\omega}^{(2)}$. Using the controlled experiment results from Section \ref{sec:case_study} we obtain $\overline{G}^{(1)} = 50.39$ avgGW and $\overline{G}^{(2)} = 50.44$ avgGW.

Consider a representative hydropower plant that can sell a fraction $\gamma \in [0,1]$ of $Q = 1$ avgMW in the forward market at price $P$, settling the remaining position in the spot market at price $\pi_{t,\omega}$. As with $g_{t,\omega}^{(1)}$ and $g_{t,\omega}^{(2)}$, we let $\pi_{t,\omega}^{(1)}$ and $\pi_{t,\omega}^{(2)}$ denote the spot prices under each policy, and $P^{(1)}$ and $P^{(2)}$ the corresponding forward prices. For $k \in \{1,2\}$, the plant's revenue in period $t$ under scenario $\omega$ is
\begin{align}
R_{t,\omega}^{(k)} \;=\; \left(P^{(k)} Q \gamma + \bigl(g_{t,\omega}^{(k)} - Q\gamma\bigr)\pi_{t,\omega}^{(k)}\right) h_t,
\label{eq:forward_revenue}
\end{align}
where $h_t$ is the number of hours in month $t$ and the plant's variable operating cost is set to zero, consistent with the low variable costs of hydropower relative to thermal plants. Under this convention, revenue coincides with profit. 
The forward price is calibrated as the ex post expectation of future spot prices under each policy, assuming a risk-neutral market with a zero risk premium\footnote{While we calibrate the forward price as the ex post expectation of spot prices under a zero-risk-premium convention, empirical evidence from the hydro-dominated Nord Pool market shows that the spot–forward relationship carries a non-zero, state-dependent risk premium that varies with hydro inflow and reservoir levels \citep{botterud2010spotfutures}; our zero-premium assumption therefore isolates the price–quantity-risk channel from premium effects, which we leave to future work.}. In this context, we assume that: (1) agents fully internalize the future price dynamics associated with each policy scenario, and (2) contracting decisions are motivated solely by hedging considerations, with contracts offering no expected gains beyond those available in the spot market. Under the biased planning policy, the resulting contract price is $P^{(1)} = 311.47$~R\$/MWh, whereas under the unbiased planning policy it is $P^{(2)} = 232.18$~R\$/MWh.

Because hydropower revenue risk is seasonal, evaluating downside risk at the monthly level may overemphasize short-term fluctuations, while pooling all revenues into a single annual distribution may hide the concentration of risk in dry periods. We therefore evaluate downside revenue risk at the quarterly level, following typical market reports and financial evaluations. 

Following the certainty-equivalent definition based on the Conditional Value-at-Risk (CVaR) in \cite{street2010conditional}, we let $\rho_\alpha(\cdot)$ denote the operator that returns the average revenue over the $\alpha$-worst-case scenarios---the lower tail of a revenue distribution. For a random revenue $R$ and tail fraction $\alpha \in (0,1)$,
\begin{align}
    \rho_\alpha(R)
    \;:=\;
    \mathbb{E}
    \left[
        R
        \,\middle|\,
        R \leq F_R^{-1}(\alpha)
    \right],
    \label{eq:cvar_operator}
\end{align}
where $F_R^{-1}(\alpha)$ is the $\alpha$-quantile of $R$, i.e., the lower-tail Value-at-Risk in the profit convention. Let $\mathcal{T}_\tau$ denote the set of months in quarter $\tau$, and define the quarterly revenue under policy $k$ and contracting level $\gamma$ as
\begin{align}
    R_{\tau,\omega}^{(k)}(\gamma)
    =
    \sum_{t \in \mathcal{T}_\tau}
    R_{t,\omega}^{(k)}(\gamma),
\end{align}
whose randomness across inflow scenarios $\omega$ defines the quarterly revenue distribution on which $\rho_\alpha$ operates.

To summarize downside contracting risk across the full simulation horizon while preserving this seasonal structure, we study the objective given by the sum of quarterly revenue CVaRs:
\begin{align}
    f^{(k)}(\gamma)
    =
    \sum_{\tau \in \Tau}
    \rho_\alpha\!\left( R_{\tau}^{(k)}(\gamma) \right),
    \label{eq:cvar_objective}
\end{align}
where $\Tau$ is the set of all quarters in the analysis. Each term $\rho_\alpha(R_{\tau}^{(k)}(\gamma))$ evaluates lower-tail revenue quarter by quarter, allowing the contracting analysis to capture the seasonal concentration of hydropower revenue risk.

Figure \ref{fig:sum_quarterly_cvar} reports the contracting strategy using the sum of quarterly revenue CVaR values for $\alpha = 5\%$. Under Policy 1, the risk-adjusted revenue reaches its maximum at a lower contracting level, $\gamma=60\%$, with a value of R\$9.10 million. Under Policy 2, the maximum occurs at $\gamma=70\%$, with a value of R\$8.18 million. Thus, in this stylized setting, the biased policy is associated with a lower optimal contracting level for a representative risk-averse hydropower producer. This is consistent with the price–quantity-risk mechanism discussed in Section \ref{sec:spot_price_results}: under Policy~1, dry-period scarcity produces sharper spot-price peaks precisely when hydropower generation is more constrained. A larger uncontracted position, therefore, has greater upside value in high-price states, while higher forward commitments increase the risk of having to settle deficits at scarcity prices. As a result, the optimal response under the biased policy is to preserve greater exposure to the spot market.

\begin{figure}[!htbp]
    \centering
    \includegraphics[scale = 0.40, trim={4.1cm 1.6cm 4.0cm 1.6cm},clip]{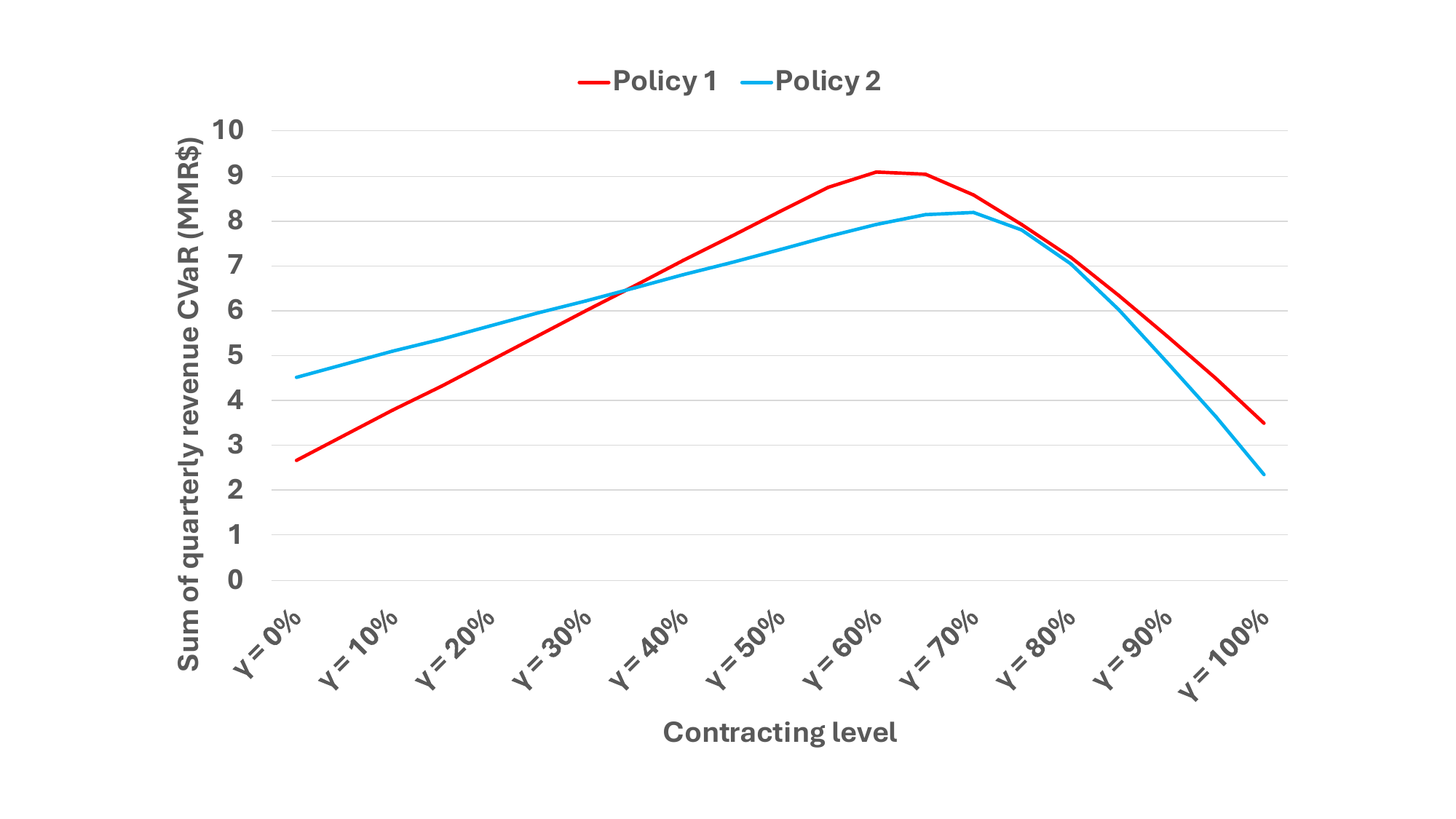}
    \caption{Sum of quarterly revenue CVaR values for different contracting levels under Policies 1 and 2.}
    \label{fig:sum_quarterly_cvar}
\end{figure}

Extrapolating the representative-unitary results to the full Brazilian hydro capacity provides a useful scale proxy for the aggregate economic magnitude of this difference. Using the respective total generation averages for each policy, $\overline{G}^{(1)} = 50.39$ and $\overline{G}^{(2)} = 50.44$ avgGW, the risk-adjusted contracting levels implied by the two policies correspond to $30.23$ avgGW under Policy~1 and $35.31$ avgGW under Policy~2. Equivalently, correcting the forecast bias---moving from Policy~1 to Policy~2---raises the willingness-to-contract by hydropower producers by approximately $16.8\%$.

These figures must be read within the assumptions of the exercise. The producer is a price taker with the average generation profile of the hydropower fleet. This is reasonable for the Brazilian case, as hydropower plants participate in the Energy Reallocation Mechanism, a compulsory hedging mechanism that allocates shares of the total hydro generation among them for market-clearing purposes, thereby providing a much more stable clearing quantity profile. The forward price equals the expected spot price under each policy \citep[see][]{Barbosa2024}. The contract is a flat quantity, as customary in the Brazilian market. The risk measure, based on the sum of quarterly revenue CVaR values, reflects typical risk profiles used in long-term contracting strategies in Brazil \citep[see][]{Fanzeres2015}. Under these assumptions, the exercise isolates the effect of the change in the joint distribution of generation and spot prices due to the inflow forecast bias on the willingness-to-contract of hydropower producers. Although this is a stylized example, the adopted setup provides a fair approximation of the decision process used by renewable generators in Brazil to define their optimal contracting level \citep[see][]{Barbosa2024}. Nevertheless, demand-side contracting, heterogeneous plants and portfolios, and risk premia are not represented. Therefore, the reported contracting levels should not be read as a forecast of actual market behavior. They indicate the direction and the order of magnitude of the incentive created by the bias. For market participants, the result means that hydropower producers have less incentive to sell forward under the biased policy due to the increased price--quantity risk, as illustrated by the October patterns of the last two years of the simulation horizon, with lower generation and higher prices, in Figures~\ref{fig:ControlledExperiment_hydro_mean} and \ref{fig:ControlledExperiment_spot}. Figure \ref{fig:sum_quarterly_cvar} also indicates that hydropower companies lack financial incentives to support a transition from the biased policy to the unbiased one. Finally, broader market equilibrium consequences of this effect are beyond the scope of this paper.

The contracting exercise indicates that the bias may also reach the contracting layer of the market. Forecast bias changes not only the timing of dispatch and spot prices, but also the joint distribution of generation and prices faced by hydropower producers. Under the biased policy, low generation is more likely to coincide with high spot prices in dry conditions, increasing price–quantity risk and reducing the hedging value of additional forward commitments. Therefore, the bias may affect not only operating costs and spot-price formation, but also the risk-management incentives of market agents.

\section{Policy implications} \label{sec:policy_implications}

First, these results have direct implications for the governance of forecasting models in hydropower-dominated systems. When inflow forecasts, SDDP cost-to-go functions, and market outcomes are tightly coupled through a centralized planning framework, persistent optimistic bias in the forecasts is not a simple statistical issue: it has the potential to depress water values, accelerate reservoir depletion, postpone thermal commitment, and translate into scarcity episodes, increasing price spikes and variability, as well as emergency out-of-merit interventions paid by the consumer. Inflow forecasting models used in official planning should therefore be evaluated not only on traditional statistical metrics but also on their ability to maintain consistency between planned and implemented policies, with explicit institutional mechanisms for transparency, independent validation, and (agile) periodic updates.

Second, the relevance of this mechanism extends beyond centralized, cost-based dispatch to competitive market settings. In many hydro-dominated systems, cost-based planning models (and the water values and prices they produce) are used as benchmarks to support decision making and regulatory studies \citep{rangel2008competition, wolfgang2009hydro}, as well as for market-monitoring purposes, providing a reference against which observed bids, prices, and dispatch can be assessed \citep{wolak2009colombia}. The bias documented here can therefore distort not only the operation itself but also these downstream benchmarking and surveillance activities. For instance, a benchmark that systematically understates scarcity will misrepresent the competitive reference and may flag efficient behavior as anomalous, or fail to detect genuinely distorted outcomes. This concern becomes even more salient under alternative market designs that have recently been studied in Brazil \citep{ccee2025meta2}. If Brazil were to move toward a bid-based market, market-monitoring procedures anchored on the existing cost-based methodology could inherit the same optimistic bias and, rather than disciplining market power, interfere with the normal functioning of the market. Ensuring that benchmark models are unbiased is thus a precondition for their legitimate use in market monitoring and oversight, independently of the dispatch paradigm in place.

Third, the analysis demonstrates the need for regulators and system operators to align reliability instruments and market signals with the behavior of planning models. When optimistic look-ahead understates scarcity and delays thermal dispatch in the face of imminent crises, short-term corrective actions and out-of-merit dispatch substitute for missing forward signals, increasing out-of-market costs and undermining the credibility of the planning framework. Regulators and system operators may therefore complement model-governance reforms with mechanisms that (i) explicitly internalize out-of-merit interventions into planning models through security constraints or scarcity pricing, (ii) monitor the gap between planned and implemented decisions to ensure that these internalized mechanisms are not themselves distorted by biased look-ahead information, and (iii) publicly disclose standardized indicators of this planning–implementation gap to promote transparency, reproducibility, and independent monitoring of planning performance.

Fourth, the bias propagation reported here does not stop at the spot-price signal. The willingness-to-contract study shows that biased forecasts can change the joint distribution of hydro generation and spot prices, increasing price–quantity risk and reducing the contracting level of a representative risk-averse hydropower producer. The governance of long-term inflow forecasts is, in this sense, also relevant to market performance: when spot prices used for settlement misrepresent scarcity, the incentives created by forward contracts may be distorted as well. Quantitatively, the contracting exercise of Section~\ref{sec:forward_market} indicates that correcting the bias raises the willingness-to-contract of a representative risk-averse hydropower producer by approximately $16.8\%$ in the stylized setting considered, which suggests that improved forecast governance may favor forward-market liquidity. However, the results also indicate that hydropower companies lack financial incentives to support a transition from the biased policy to the unbiased one. As the forward market represents the main competitive arena in the Brazilian electricity market, as well as in most Latin American electricity markets, market operators should monitor and publicly disclose standardized inconsistency indices and ensure that the planning models underlying spot-price formation remain empirically validated and free from systematic forecast bias, thereby preserving efficient market signals, forward-market liquidity, and price discovery. 

Fifth, the distortions documented here are not transient deviations but the standing difference between two long-run operating regimes. Storage is a stock variable that accumulates each period's dispatch decisions, so the small monthly over-discharge induced by the bias (modest in isolation) builds up over time into the large differences in storage, prices, costs, and reliability documented in Sections~\ref{sec:empirical_evidence}--\ref{sec:forward_market}. In the Brazilian case, this bias has persisted in the official process for more than a decade, so these differences characterize a long-standing operating regime rather than a temporary departure from it. Correcting the bias offers a real long-run gain, on the order of the cost and reliability differences in Table~\ref{tab:ControlledExperiment_costs}, but the move from the biased to the unbiased regime is itself an economic event: it raises spot prices on impact and keeps them higher during wet seasons, precisely when the inconsistent policy artificially signals abundance, in order to halt the reservoir-depletion trend. This repricing creates winners and losers, as short-term-exposed consumers pay more in normal periods while the system gains protection against dry-season scarcity, and it shifts the willingness-to-contract between regimes (Section~\ref{sec:forward_market}). Prices, storage, and contract positions therefore pass through a transient before settling into the unbiased regime; characterizing this transient---its duration, welfare incidence, and consequences for contract design---is beyond the scope of this paper and constitutes a relevant topic for future research.

\section{Conclusion}

This paper examined whether systematic optimistic inflow-forecast bias in a storage-rich hydrothermal power system can distort water values, system operation, and market incentives. We tested this mechanism using Brazil’s official hydro-scheduling models and data. Within the limitations of this study, including its modeling choices, data, and assumptions, the results provide strong evidence that the two research questions raised at the beginning of this paper can be answered in the affirmative for the case analyzed, demonstrating that this mechanism can arise and propagate from planning forecasts to long-term market incentives and, ultimately, long-term system reliability. All data used in this study are available in \cite{Brigatto2026BiasedInflowsData}.

While the theoretical analysis shows, in a stylized setting, that optimistic future inflow forecasts reduce the perceived opportunity cost of stored water and induce greater hydro use than would arise under unbiased forecasts, the empirical and experimental evidence traces this distortion through the successive layers of the system. The depressed marginal valuation of storage advances hydropower discharge and accelerates reservoir drawdown; the resulting depletion postpones the preventive commitment of thermal capacity and concentrates it into sharper dry-season dispatch; and the merit-order substitution that artificially lowers the wet-season price signal magnifies the scarcity rents realized under dry conditions. The Brazilian operation record is consistent with this chain: planned hydro generation and storage trajectories systematically exceed their implemented counterparts, while thermal generation, spot prices, and out-of-merit dispatch reveal corrective action deferred until realized inflows fall short of expectations. The controlled SDDP experiments isolate and quantify the counterfactual: relative to the bias-corrected policy, the biased policy yields lower reservoir levels, sharper dry-season thermal dispatch, higher dry-season spot-price expected returns, a 10.3\% load-shedding probability, and substantially higher expected operating costs.

The forward-market analysis shows that the distortion may extend beyond physical operation and spot-price formation. In the stylized contracting overlay, the biased policy changes the joint distribution of hydropower generation and spot prices, increasing price–quantity risk: low generation becomes more likely to coincide with scarcity prices in dry conditions. This reduces the hedging value of forward commitments for a risk-averse hydropower producer and lowers its optimal contracting level. Thus, the same forecast bias that distorts water values may also affect the contractual layer organized around the spot-price signal. These findings indicate that forecast-bias monitoring and correction are central to the governance of SDDP-based hydrothermal planning systems. 

Although the evidence presented in this paper is based on the Brazilian system, the analytical result of Section~\ref{sec:implications} does not depend on its specific features. It is mainly based on the assumption that the system operator determines its first-stage decisions with a long-term multistage operation planning model fed by inflow forecasts. In this sense, the main structure underlying the mechanism is a planning model with a biased look-ahead. The mechanism can therefore be expected to arise under three main conditions. First, dispatch and storage decisions depend on forecast-based estimates of the opportunity cost of stored energy. Second, the optimistic bias is sufficiently large, and the forecasting information is either shared across agents or centrally defined. Third, storage capacity is large enough for the management of stored energy to considerably affect operating costs, but not so abundant that its opportunity cost is zero most of the time. How the mechanism manifests in practice, and the magnitude of its effects, depend on the characteristics of each system.

Future research may extend this work by jointly treating inflow bias and other structural simplifications in planning models, and integrating external information and statistical corrections directly into the SDDP loop. Additionally, over the last decade, the discrepancies documented in Section~\ref{sec:empirical_evidence} have raised concerns, as discussed in the Introduction, leading the system operator to increase the model’s risk aversion as an empirical compensation mechanism. The annual calibration of this risk-aversion parameter relies on ad hoc procedures that are beyond the scope of this paper. Assessing how the effects identified in the controlled experiments vary with the level of risk aversion is an important direction for future research.

\section*{Acknowledgements}

The authors would like to thank Frank Wolak for his valuable suggestions regarding the empirical studies and Sérgio Granville for his careful review of the theoretical proofs.

\section*{Declaration of Generative AI and AI-assisted technologies in the writing process}

During the preparation of this work the authors used both ChatGPT 5.5 and Claude Opus 4.8 interchangeably in order to improve language and readability. After using these tools, the authors reviewed and edited the content as needed and take full responsibility for the content of the publication.

\appendix 

\section{Proofs} \label{appendix:proofs}

This appendix collects the full proofs of Proposition~\ref{prop:bias_reduces_water_value}, Theorem~\ref{theo:bias_increase_discharge}, and Theorem~\ref{theo:theorem_stochastic}. The proofs work with the dynamic-programming recursion stated in Section~\ref{sec:implications}: for each period $t$, given the storage state $v_{t-1}$ and a forecast bias vector $\boldsymbol{\delta}=(\delta_2,\dots,\delta_T)\geq\boldsymbol{0}$ (with $\delta_1=0$ by convention), the cost-to-go function $Q_t^{(\boldsymbol{\delta})}$ is defined by~\eqref{eq:reduced_stage_problem_obj}--\eqref{eq:reduced_stage_problem_storage_bounds}, and the immediate cost $\phi_t(r_t)$, written as a function of the total water release $r_t$, is defined by~\eqref{eq:phi_definition_obj}--\eqref{eq:phi_definition_turbine_bound}. Throughout, $\overline V$ and $\underline V$ denote the storage bounds, $\overline U$ the turbine upper bound, $\boldsymbol{c}_t$ the vector of thermal marginal costs, and $\boldsymbol{e}_{t_0}$ the standard basis vector with a one in position $t_0$. The set ordering $\succeq$ between subdifferential intervals is the endpoint ordering defined in the footnote of Proposition~\ref{prop:bias_reduces_water_value}: for nonempty closed intervals $A,B\subseteq\overline{\mathbb R}$, $A\succeq B$ iff $\inf A\ge \inf B$ and $\sup A\ge \sup B$. %

Throughout the appendix we maintain the following standing assumptions, which hold in the hydrothermal setting analyzed here. For every period $t$: (i) inflows are nonnegative, $y_t\ge 0$ (almost surely, in the stochastic setting of Theorem~\ref{theo:theorem_stochastic}); (ii) the available thermal capacity is sufficient to ensure power supply, $\boldsymbol{1}^\top\underline{\boldsymbol G}\le d_t\le \boldsymbol{1}^\top\overline{\boldsymbol G}$, which holds automatically when a finite-cost deficit unit is included; (iii) marginal costs are nonnegative, $\boldsymbol{c}_t\ge\boldsymbol{0}$; and (iv) the analysis adopts the free-spill specialization $\mathcal P=\{(u_t,s_t): 0\le u_t\le\overline U,\ s_t\ge 0\}$, so spillage disposes of any surplus inflow at zero cost. Under (i)--(iv), each $\phi_t$ is finite, piecewise-linear convex, and non-increasing on $\mathbb R_+$, and each $Q_t^{(\boldsymbol{\delta})}$ is finite, convex, and non-increasing on $[\underline V,\infty)$.

\subsection*{Proof of Proposition~\ref{prop:bias_reduces_water_value}}

\begin{proof}
We argue by backward induction on $t$, from $t=t_0$ down to $t=1$. At each step the inductive hypothesis is the subgradient ordering~\eqref{eq:prop_subgrad_ordering} at stage $t+1$. The induction step proves~\eqref{eq:prop_u_ordering} at stage $t$ (for every $t=1,\dots,t_0$) and, when $t\geq 2$, also~\eqref{eq:prop_subgrad_ordering} at stage $t$.

\medskip
\noindent\textbf{Base case ($t=t_0$).} For periods $\tau>t_0$, the two bias vectors coincide, hence $Q_{t_0+1}^{(\boldsymbol{\delta}')}=Q_{t_0+1}^{(\boldsymbol{\delta})}$. By construction $\delta_{t_0}=0$ and $\delta'_{t_0}=\eta$, so the two stage-$t_0$ recursions in~\eqref{eq:reduced_stage_problem_obj}--\eqref{eq:reduced_stage_problem_storage_bounds} differ only by the inflow term in the water balance~\eqref{eq:reduced_stage_problem_balance}, which is $y_{t_0}$ under $\boldsymbol{\delta}$ and $y_{t_0}+\eta$ under $\boldsymbol{\delta}'$. Substituting $v_{t_0}=(v_{t_0-1}+\eta)+y_{t_0}-r_{t_0}$ shows that the $\boldsymbol{\delta}'$ problem at state $v_{t_0-1}$ coincides with the $\boldsymbol{\delta}$ problem at the shifted state $v_{t_0-1}+\eta$, since the storage bound~\eqref{eq:reduced_stage_problem_storage_bounds} restricts only the post-state $v_{t_0}$ and the inner subproblem~\eqref{eq:phi_definition_obj}--\eqref{eq:phi_definition_turbine_bound} can absorb any surplus through spillage. Therefore,
\begin{align}
Q_{t_0}^{(\boldsymbol{\delta}')}(v)\;=\;Q_{t_0}^{(\boldsymbol{\delta})}(v+\eta),\qquad \forall\,v\in[\underline V,\overline V].
\label{eq:base_shift_identity}
\end{align}
The pre-state $v_{t_0-1}$ enters the stage-$t_0$ problem only through the right-hand side of the balance~\eqref{eq:reduced_stage_problem_balance}, and the problem is feasible for every $v_{t_0-1}\ge\underline V$ because the spillage variable in~\eqref{eq:phi_definition_release} absorbs any surplus inflow at zero cost. Hence $Q_{t_0}^{(\boldsymbol{\delta})}$ is finite, convex, and non-increasing on $[\underline V,\infty)$, and identity~\eqref{eq:base_shift_identity} holds for every $v\in[\underline V,\overline V]$, including arguments with $v+\eta>\overline V$.

The cost-to-go function is convex, hence its subdifferential is monotone non-decreasing in the endpoint sense \citep{rockafellar1997convex}. Together with $\eta\geq 0$,~\eqref{eq:base_shift_identity} delivers
\begin{align}
\partial_v Q_{t_0}^{(\boldsymbol{\delta}')}(v)
\;=\;\partial_v Q_{t_0}^{(\boldsymbol{\delta})}(v+\eta)
\;\succeq\;
\partial_v Q_{t_0}^{(\boldsymbol{\delta})}(v),
\end{align}
which is~\eqref{eq:prop_subgrad_ordering} at $t=t_0$. The discharge ordering~\eqref{eq:prop_u_ordering} at $t=t_0$ follows from the shift identity~\eqref{eq:base_shift_identity}: the greatest optimal release is non-decreasing in the pre-state, so $r_{t_0}^{*(\boldsymbol{\delta}')}(v)=r_{t_0}^{*(\boldsymbol{\delta})}(v+\eta)\ge r_{t_0}^{*(\boldsymbol{\delta})}(v)$, and in~\eqref{eq:phi_definition_obj}--\eqref{eq:phi_definition_turbine_bound} the greatest optimal discharge is the greatest minimizer of the inner problem over the growing feasible interval $[0,\min(r_{t_0},\overline U)]$, hence non-decreasing in $r_{t_0}$; therefore $u_{t_0}^{*(\boldsymbol{\delta}')}\ge u_{t_0}^{*(\boldsymbol{\delta})}$. As a consequence of the release ordering $r_{t_0}^{*(\boldsymbol{\delta}')}\ge r_{t_0}^{*(\boldsymbol{\delta})}$ and the water balance equation, it is straightforward to see that $v_{t_0}^{*(\boldsymbol{\delta}',\boldsymbol{\delta})}\le v_{t_0}^{*(\boldsymbol{\delta},\boldsymbol{\delta})}$. Finally, for completeness, the same rationale applies to the lowest minimizer. Note that the set of minimizers of a convex linear program is a closed and convex set.   

\medskip
\noindent\textbf{Induction step ($1\leq t<t_0$).} Suppose~\eqref{eq:prop_subgrad_ordering} holds at stage $t+1$. For finite convex functions $F,G$ on $[\underline V,\overline V]$, the endpoint ordering $\partial F(v)\succeq \partial G(v)$ at every $v$ is equivalent to $F-G$ being non-decreasing on $[\underline V,\overline V]$: indeed, let $\partial F(v)=[F'_-(v),F'_+(v)]$ and likewise for $G$, the ordering reads $F'_-(v)\geq G'_-(v)$ and $F'_+(v)\geq G'_+(v)$, which is precisely $(F-G)'_\pm(v)\geq 0$ at every $v$. Applying this equivalence at stage $t+1$,
\begin{align}
\psi(v)\;:=\;Q_{t+1}^{(\boldsymbol{\delta}')}(v)-Q_{t+1}^{(\boldsymbol{\delta})}(v)\quad\text{is non-decreasing on $[\underline V,\overline V]$.}
\label{eq:psi_nondecreasing}
\end{align}

For $t<t_0$ we have $\delta_t=\delta'_t=0$, so the two stage-$t$ recursions differ only in the cost-to-go function. Eliminating $r_t=v_{t-1}+y_t-v_t$ from the water balance, define
\begin{align}
H^{(i)}(v_t;v_{t-1})\;:=\;\phi_t(v_{t-1}+y_t-v_t)+Q_{t+1}^{(i)}(v_t),\qquad i\in\{\boldsymbol{\delta},\boldsymbol{\delta}'\},
\end{align}
which is convex in $v_t$ on the closed feasible interval
\(
\mathcal V(v_{t-1}):=[\underline V,\overline V]\cap(-\infty,v_{t-1}+y_t]
\)
(the upper limit encodes $r_t\geq 0$). By construction,
\(
H^{(\boldsymbol{\delta}')}(v_t;v_{t-1})-H^{(\boldsymbol{\delta})}(v_t;v_{t-1})=\psi(v_t),
\)
which is non-decreasing in $v_t$.

\medskip
\noindent\emph{Step 1: hydro-discharge ordering~\eqref{eq:prop_u_ordering} at stage $t$.}
Let $v_t^{*(i)}=\inf[\arg\min_{v_t\in\mathcal V(v_{t-1})}H^{(i)}(v_t;v_{t-1})]$, and suppose, for contradiction, that $v_t^{*(\boldsymbol{\delta}')}>v_t^{*(\boldsymbol{\delta})}$. Adding the two optimality inequalities,
\begin{align}
H^{(\boldsymbol{\delta}')}(v_t^{*(\boldsymbol{\delta}')})\le H^{(\boldsymbol{\delta}')}(v_t^{*(\boldsymbol{\delta})}),\qquad
H^{(\boldsymbol{\delta})}(v_t^{*(\boldsymbol{\delta})})\le H^{(\boldsymbol{\delta})}(v_t^{*(\boldsymbol{\delta}')}),
\end{align}
and rearranging, $\psi(v_t^{*(\boldsymbol{\delta}')})\le \psi(v_t^{*(\boldsymbol{\delta})})$. Combining the last inequality with~\eqref{eq:psi_nondecreasing} and the hypothesis $v_t^{*(\boldsymbol{\delta}')}>v_t^{*(\boldsymbol{\delta})}$, then $\psi(v_t^{*(\boldsymbol{\delta}')}) = \psi(v_t^{*(\boldsymbol{\delta})})$, i.e., $\psi$ is constant on $[v_t^{*(\boldsymbol{\delta})},v_t^{*(\boldsymbol{\delta}')}]$. Thus, $v_t^{*(\boldsymbol{\delta})}$ is itself optimal for $H^{(\boldsymbol{\delta}')}$.\footnote{To keep the exposition concise, we leave the proof of this step to the reader.} Therefore, the lowest optimal post-state $v_t^{*}$ is non-increasing in the bias, violating the step 1 assumption that $v_t^{*(\boldsymbol{\delta}')}>v_t^{*(\boldsymbol{\delta})}$.

The water balance gives $r_t^{*(\boldsymbol{\delta}')}\ge r_t^{*(\boldsymbol{\delta})}$. Within the inner subproblem~\eqref{eq:phi_definition_obj}--\eqref{eq:phi_definition_turbine_bound} that defines $\phi_t$, the greatest optimal discharge is the greatest minimizer of the fixed convex inner objective over the growing feasible interval $[0,\min(r_t,\overline U)]$, hence non-decreasing in $r_t$ (which also shows the optimal discharges form a closed interval, as used in part~(ii)). Hence the greatest optimal discharge $u_t^{*}$ is non-decreasing in the bias, which is~\eqref{eq:prop_u_ordering} for the selection $u_t^{*}$ of part~(ii) at stage $t$.

\medskip
\noindent\emph{Step 2: subgradient ordering~\eqref{eq:prop_subgrad_ordering} at stage $t$.} For $t=1$ there is no subgradient ordering to prove, so suppose $t\geq 2$ and define, for $i\in\{\boldsymbol{\delta},\boldsymbol{\delta}'\}$,
\begin{align}
A^{(i)}\;:=\;\partial \phi_t(r_t^{*(i)}),
\qquad
B^{(i)}\;:=\;\partial Q_{t+1}^{(i)}(v_t^{*(i)}),
\qquad
D^{(i)}\;:=\;
\begin{cases}
(-\infty,\max B^{(i)}], & v_t^{*(i)}=\underline V,\\
B^{(i)}, & \underline V<v_t^{*(i)}<\overline V,\\
[\min B^{(i)},\infty), & v_t^{*(i)}=\overline V.
\end{cases}
\end{align}
The Lagrangian of~\eqref{eq:reduced_stage_problem_obj}--\eqref{eq:reduced_stage_problem_storage_bounds} is
$\mathcal L^{(i)}\;=\;\phi_t(r_t)+Q_{t+1}^{(i)}(v_t)+\lambda_t\,(v_t-v_{t-1}-y_t-\delta_t+r_t)+\underline\mu_t(\underline V-v_t)+\overline\mu_t(v_t-\overline V)$,

with $\underline\mu_t,\overline\mu_t\ge 0$, $\lambda_t\in\mathbb R$, and $\delta_t=0$ throughout the induction step since $t<t_0$. The stage problem is a polyhedral convex program ($\phi_t$ and $Q_{t+1}^{(i)}$ are value functions of linear programs, hence polyhedral) with finite optimal value under the standing assumptions, so strong duality holds and the set $\Lambda^*(v_{t-1})$ of optimal multipliers in the balance constraint is nonempty \citep[Cor.~28.2.2]{rockafellar1997convex}. 

We want to show that $-\Lambda^*(v_{t-1})=A^{(i)}\cap D^{(i)}$. To do that, we check this equality by showing that the elements of each set belong to the other. So, let $(r_t^{*(i)},v_t^{*(i)},\lambda_t^{*(i)},\underline\mu_t^{*(i)},\overline\mu_t^{*(i)})$ be any saddle point of $\mathcal L^{(i)}$. KKT stationarity in $r_t$ gives $-\lambda_t^{*(i)}\in A^{(i)}$; stationarity in $v_t$, combined with complementary slackness on the storage bounds, gives $-\lambda_t^{*(i)}\in D^{(i)}$, by case-split on whether $v_t^{*(i)}$ lies on $\underline V$, in the interior, or on $\overline V$. Conversely, fix any $\mu\in A^{(i)}\cap D^{(i)}$ and set $\lambda_t^{*(i)}:=-\mu$. Choose the bound multipliers as $\underline\mu_t^{*(i)}=\overline\mu_t^{*(i)}=0$ if $v_t^{*(i)}$ is interior; if $v_t^{*(i)}=\underline V$, set $\overline\mu_t^{*(i)}=0$ and $\underline\mu_t^{*(i)}=b-\mu\ge 0$ for some $b\in B^{(i)}$ with $b\ge\mu$ (such a $b$ exists because $\mu\in D^{(i)}=(-\infty,\max B^{(i)}]$); the case $v_t^{*(i)}=\overline V$ is symmetric. The resulting tuple is a saddle point of $\mathcal L^{(i)}$, hence $\lambda_t^{*(i)}\in\Lambda^*(v_{t-1})$. Together, the two directions give $-\Lambda^*(v_{t-1})=A^{(i)}\cap D^{(i)}$. 

By the convex-perturbation envelope theorem \citep[Thm.~29.1]{rockafellar1997convex}, the subdifferential of the optimal-value function in the right-hand-side parameter $v_{t-1}$ equals $-\Lambda^*(v_{t-1})$, since $v_{t-1}$ enters $\mathcal L^{(i)}$ linearly through the balance with coefficient $-\lambda_t$. Hence $\partial_{v_{t-1}} Q_t^{(i)}(v_{t-1})\;=\;A^{(i)}\cap D^{(i)},\;\; i\in\{\boldsymbol{\delta},\boldsymbol{\delta}'\}$.

Notice that $\Lambda^*(v_{t-1})$ depends on the parameter $v_{t-1}$ alone and not on the choice of primal optimum, since in a convex program with strong duality every primal optimum pairs with every dual optimum at a saddle point \citep[Thm.~28.3]{rockafellar1997convex}. The right-hand side of the previous equality is therefore the same regardless of which optimal $(r_t^{*(i)},v_t^{*(i)})$ is used to compute $A^{(i)}$ and $D^{(i)}$. In particular $A^{(i)}\cap D^{(i)}\neq\emptyset$, inheriting nonemptiness from $\Lambda^*(v_{t-1})\neq\emptyset$.

We now compare $A^{(\boldsymbol{\delta}')}\cap D^{(\boldsymbol{\delta}')}$ with $A^{(\boldsymbol{\delta})}\cap D^{(\boldsymbol{\delta})}$ in the two cases left open by Step~1.

\medskip
\noindent\emph{Case A: $r_t^{*(\boldsymbol{\delta}')}>r_t^{*(\boldsymbol{\delta})}$.} Since $\phi_t$ is convex, its subdifferential is monotone non-decreasing, hence $\min A^{(\boldsymbol{\delta}')}\ge \max A^{(\boldsymbol{\delta})}$. Therefore every element of $A^{(\boldsymbol{\delta}')}\cap D^{(\boldsymbol{\delta}')}$ is at least as large as every element of $A^{(\boldsymbol{\delta})}\cap D^{(\boldsymbol{\delta})}$; in particular,
$\min\bigl(A^{(\boldsymbol{\delta}')}\cap D^{(\boldsymbol{\delta}')}\bigr)
\;\ge\;
\max\bigl(A^{(\boldsymbol{\delta})}\cap D^{(\boldsymbol{\delta})}\bigr)$, which a fortiori delivers~\eqref{eq:prop_subgrad_ordering} at stage $t$.

\medskip
\noindent\emph{Case B: $r_t^{*(\boldsymbol{\delta}')}=r_t^{*(\boldsymbol{\delta})}=:r^*$.} The water balance gives $v_t^{*(\boldsymbol{\delta}')}=v_t^{*(\boldsymbol{\delta})}=:v^*$, so $A^{(\boldsymbol{\delta}')}=A^{(\boldsymbol{\delta})}=:A$. The inductive hypothesis~\eqref{eq:prop_subgrad_ordering} at stage $t+1$, evaluated at $v^*$, gives: $\min B^{(\boldsymbol{\delta}')}\ge \min B^{(\boldsymbol{\delta})}$ and $\max B^{(\boldsymbol{\delta}')}\ge \max B^{(\boldsymbol{\delta})}$.

We examine the two endpoints of $A\cap D^{(i)}$ in each sub-case of $v^*$:

\begin{itemize}
\item If $\underline V<v^*<\overline V$, then $D^{(i)}=B^{(i)}$ and
$
A\cap B^{(i)}=\bigl[\max(\min A,\min B^{(i)}),\;\min(\max A,\max B^{(i)})\bigr].
$
Both endpoints are non-decreasing in $\min B^{(i)}$ and $\max B^{(i)}$, so $A\cap B^{(\boldsymbol{\delta}')}\succeq A\cap B^{(\boldsymbol{\delta})}$.

\item If $v^*=\underline V$, then $D^{(i)}=(-\infty,\max B^{(i)}]$ and
$
A\cap D^{(i)}=\bigl[\min A,\;\min(\max A,\max B^{(i)})\bigr].
$
The lower endpoint $\min A$ is the same for both $i$; the upper endpoint is non-decreasing in $\max B^{(i)}$. Hence $A\cap D^{(\boldsymbol{\delta}')}\succeq A\cap D^{(\boldsymbol{\delta})}$.

\item If $v^*=\overline V$, then $D^{(i)}=[\min B^{(i)},\infty)$ and
$
A\cap D^{(i)}=\bigl[\max(\min A,\min B^{(i)}),\;\max A\bigr].
$
The upper endpoint $\max A$ is the same for both $i$; the lower endpoint is non-decreasing in $\min B^{(i)}$. Hence $A\cap D^{(\boldsymbol{\delta}')}\succeq A\cap D^{(\boldsymbol{\delta})}$.
\end{itemize}

In every case, $\partial_{v_{t-1}} Q_t^{(\boldsymbol{\delta}')}(v_{t-1})\succeq \partial_{v_{t-1}} Q_t^{(\boldsymbol{\delta})}(v_{t-1})$, which is~\eqref{eq:prop_subgrad_ordering} at stage $t$. The induction step is complete.

\medskip
By backward induction,~\eqref{eq:prop_subgrad_ordering} holds for every $t=2,\dots,t_0$ and~\eqref{eq:prop_u_ordering} holds for every $t=1,\dots,t_0$.
\end{proof}

\subsection*{Proof of Theorem~\ref{theo:bias_increase_discharge}}

\begin{proof}
We build a sequence of bias vectors by adding the components of $\boldsymbol{\delta}$ one at a time, from period $T$ backward to period $2$. Let $\boldsymbol{\delta}^{(T+1)}:=\boldsymbol{0}$ and, for $k=T,T-1,\dots,2$,
$\boldsymbol{\delta}^{(k)}\;:=\;\boldsymbol{\delta}^{(k+1)}+\delta_k\,\boldsymbol{e}_k\;=\;(0,\dots,0,\delta_k,\delta_{k+1},\dots,\delta_T)$,
so that $\boldsymbol{\delta}^{(2)}=\boldsymbol{\delta}$. By construction, each pair $\bigl(\boldsymbol{\delta}^{(k+1)},\boldsymbol{\delta}^{(k)}\bigr)$ satisfies the hypothesis of Proposition~\ref{prop:bias_reduces_water_value} with $t_0=k$ and $\eta=\delta_k$: the vector $\boldsymbol{\delta}^{(k+1)}$ has zero entries in periods $1,\dots,k$, and $\boldsymbol{\delta}^{(k)}$ adds $\delta_k\geq 0$ at period $k$.

Applying part~(ii) of Proposition~\ref{prop:bias_reduces_water_value} at $t=1$ for each $k=T,T-1,\dots,2$, $u_1^{*(\boldsymbol{\delta}^{(k)})}\;\geq\;u_1^{*(\boldsymbol{\delta}^{(k+1)})}$.
Chaining these inequalities,
$u_1^{*(\boldsymbol{\delta})}
\;=\;u_1^{*(\boldsymbol{\delta}^{(2)})}
\;\geq\;u_1^{*(\boldsymbol{\delta}^{(3)})}
\;\geq\;\cdots
\;\geq\;u_1^{*(\boldsymbol{\delta}^{(T)})}
\;\geq\;u_1^{*(\boldsymbol{\delta}^{(T+1)})}
\;=\;u_1^{*(\boldsymbol{0})}$.

\end{proof}

\subsection*{Proof of Theorem~\ref{theo:theorem_stochastic}}

The proof extends the backward induction of Proposition~\ref{prop:bias_reduces_water_value} to the expected cost-to-go, applying the stage-$t$ argument conditionally on the innovation $\boldsymbol{\epsilon}_t$ and taking expectations at each stage.

\begin{proof}
We argue by backward induction on $t$ and then add the bias components one at a time as in Theorem~\ref{theo:bias_increase_discharge}. We show that for any single-component pair $\boldsymbol\delta'=\boldsymbol\delta+\eta\,\boldsymbol e_{t_0}$ ($\eta\ge0$) the difference $Q_t^{(\boldsymbol\delta')}-Q_t^{(\boldsymbol\delta)}$ is non-decreasing in $v_{t-1}$, equivalently $\partial_v Q_t^{(\boldsymbol\delta')}\succeq\partial_v Q_t^{(\boldsymbol\delta)}$.

\emph{Single-component step.} Conditional on $\epsilon_t$, the stage-$t$ recourse problem
\begin{align}
q_t^{(i)}(v_{t-1},\epsilon_t)=\min_{r_t,v_t}\Big\{\phi_t(r_t)+Q_{t+1}^{(i)}(v_t):\ v_t=v_{t-1}+\hat{y}_t+\epsilon_t+\delta_t^{(i)}-r_t,\ \underline V\le v_t\le\overline V\Big\},\quad i\in\{\boldsymbol\delta,\boldsymbol\delta'\},
\end{align}
is a deterministic single-stage problem whose cost-to-go function $Q_{t+1}^{(i)}$ is convex and non-increasing, being an expectation of such functions. The induction step of Proposition~\ref{prop:bias_reduces_water_value} (Steps~1 and~2) uses convexity and finiteness of the cost-to-go function, the non-decreasing-difference hypothesis at $t+1$, and existence of an optimal multiplier for the conditional stage problem; for every state with $v_{t-1}+\hat{y}_t+\epsilon_t+\delta_t>\underline V$ the feasible set meets the relative interior of the domain of the objective, so the multiplier exists per realization by the relative-interior constraint qualification, while at the corner state $v_{t-1}+\hat{y}_t+\epsilon_t+\delta_t=\underline V$ the feasible set is the singleton $(r_t,v_t)=(0,\underline V)$ and the ordering follows by one-sided limits of the finite convex functions involved. Applied for each realization $\epsilon_t$, it gives that $q_t^{(\delta')}-q_t^{(\delta)}$ is non-decreasing in $v_{t-1}$ and $u_t^{*(\delta')}\ge u_t^{*(\delta)}$ for $t<t_0$, while at the added stage $t=t_0$ the base case of Proposition~\ref{prop:bias_reduces_water_value} (inflow shift $\eta$) gives the same two conclusions. Under the standing assumptions, $0\le q_t\le\sum_{\tau\ge t}\boldsymbol c_\tau^\top\overline{\boldsymbol G}$ is uniformly bounded, so all expectations below are finite. Since expectation preserves monotonicity,
$Q_t^{(\boldsymbol\delta')}-Q_t^{(\boldsymbol\delta)}=\mathbb E_{\epsilon_t}\!\big[q_t^{(\boldsymbol\delta')}-q_t^{(\boldsymbol\delta)}\big]$
is non-decreasing in $v_{t-1}$, which by the equivalence in the proof of Proposition~\ref{prop:bias_reduces_water_value} (valid for any finite convex functions, and $Q_t^{(i)}$ is convex as an expectation of convex functions) is $\partial_v Q_t^{(\boldsymbol\delta')}\succeq\partial_v Q_t^{(\boldsymbol\delta)}$. For $t>t_0$ the two bias vectors coincide, so $Q_t^{(\boldsymbol\delta')}=Q_t^{(\boldsymbol\delta)}$ and the ordering holds with equality; the inequality is established at stages $t\le t_0$. The leaf $Q_{T+1}^{(i)}\equiv0$ starts the backward induction.

\emph{Combining the single-component steps.} Writing $\boldsymbol\delta$ as the sum of single-component additions from period $T$ down to $2$ as in Theorem~\ref{theo:bias_increase_discharge}, the single-component step applies to each consecutive pair and yields $\partial_v Q_t^{(\boldsymbol\delta)}\succeq\partial_v Q_t^{(\boldsymbol0)}$ for $t=2,\dots,T$. Chaining the orderings of the greatest optimal first-stage discharge gives $u_1^{*(\boldsymbol\delta)}\ge u_1^{*(\boldsymbol0)}$.
\end{proof}

\section{Hydrothermal power systems operation planning} \label{sec:appendix_sddp}

For completeness, this appendix presents the hydrothermal operation-planning
model in a compact formulation similar to previous works \cite{street2020assessing}. Relative to the stylized model of Section \ref{sec:implications}, boldface denotes the corresponding system-wide vectors as multiple plants and transmission lines can be present in the full formulation. The conditional stage problem is
defined as:
\begin{align}
&q_t(\boldsymbol{v}_{t-1},\boldsymbol{y}_{[t-1]},
     \boldsymbol{\epsilon}_{t\omega}) =
\min_{\boldsymbol{v}_t,\boldsymbol{g}_t,\boldsymbol{u}_t,
      \boldsymbol{s}_t,\boldsymbol{f}_t}
\left\{
\boldsymbol{c}_t^\top\boldsymbol{g}_t
+JQ_{t+1}(\boldsymbol{v}_t,\boldsymbol{y}_{[t]})
\;\middle|\;
\begin{alignedat}{2}
&\boldsymbol{A}_t\boldsymbol{f}_t
+\boldsymbol{B}_t\boldsymbol{g}_t
+\boldsymbol{P}_t\boldsymbol{u}_t
=\boldsymbol{d}_t,
&\qquad&
\\
&\boldsymbol{v}_t
=\boldsymbol{v}_{t-1}+\boldsymbol{y}_t
-\boldsymbol{H}_t(\boldsymbol{u}_t+\boldsymbol{s}_t),
&& 
\\
&\boldsymbol{y}_t
=\boldsymbol{M}_t\boldsymbol{y}_{[t-1]}
+\boldsymbol{N}_t\boldsymbol{\epsilon}_{t\omega},
&&
\\
&(\boldsymbol{v}_t,\boldsymbol{g}_t,\boldsymbol{u}_t,
 \boldsymbol{s}_t,\boldsymbol{f}_t)
\in\mathcal{X}_t,
\end{alignedat}
\right\}.
\end{align}

The decision vectors $\boldsymbol{g}_t$, $\boldsymbol{u}_t$, $\boldsymbol{s}_t$, $\boldsymbol{f}_t$, and $\boldsymbol{v}_t$ represent thermal generation,
hydro discharge, spillage, transmission flows, and end-of-period reservoir storage, respectively. The vectors $\boldsymbol{y}_t$ and $\boldsymbol{y}_{[t]}$ contain the current inflow and the inflow information carried between stages, while $\boldsymbol{\epsilon}_{t\omega}$ is the inflow innovation in realization $\omega$. The first constraint represents the energy balance, the second
represents the cascade water balance, and the third defines the linear stochastic inflow process and determines $\boldsymbol{y}_t$. The associated matrices map generation and flows to the electrical network,
represent the hydro cascade topology, and propagate the inflow state an innovations. The final constraint collects the remaining operational constraints and variable bounds in the set $\mathcal{X}_t$, while $J$ is
the inter-stage discount factor. The expected cost-to-go function is $Q_{t+1}\left(
\boldsymbol{v}_{t},
\boldsymbol{y}_{[t]}
\right)
=
\mathbb{E}_{\boldsymbol{\epsilon}_{t+1}}
\left[
q_{t+1}\left(
\boldsymbol{v}_{t},
\boldsymbol{y}_{[t]},
\boldsymbol{\epsilon}_{t+1}
\right)
\right]
= \sum_{\omega\in\Omega_{t+1}}
\mathbb{P}_\omega\,
q_{t+1}\left(
\boldsymbol{v}_{t},
\boldsymbol{y}_{[t]},
\boldsymbol{\epsilon}_{t+1,\omega}
\right)$,
with terminal condition $Q_{T+1}\equiv0$. The application of SDDP, including the piecewise-linear approximation of the cost-to-go functions through Benders cuts and the forward and backward passes of the algorithm, follows the approach of \cite{Dowson2021SDDPjl}.

\section{Sensitivity analysis} \label{sec:appendix_sensitivity}

This appendix repeats the controlled experiments of Section~\ref{sec:case_study} for four additional cases. In the first two, the initial storage of every hydro reservoir is set to 90\% and to 10\% of its maximum capacity. In the other two, the capacity of every thermal plant is increased and reduced by 20\%. In each case, Policies 1 and 2 are trained again with the same protocol as the base case and evaluated under the same 2,000 out-of-sample scenarios. Table~\ref{tab:sensitivity} reports, for each case and policy, the following operating and market performance indices: the statistics of Table~\ref{tab:ControlledExperiment_costs} and the five-year average of the spot-price expected returns.

\begin{table}[H]
\centering
\caption{Sensitivity of the operating and market performance indices to the initial storage level and to the thermal capacity. Costs in MMR\$, as in Table~\ref{tab:ControlledExperiment_costs}.}
\label{tab:sensitivity}
\resizebox{\textwidth}{!}{%
\begin{tabular}{llccccccc}
\hline
Case & Policy & P5\% & \shortstack[c]{\rule{0pt}{2.6ex}Sample\\average} & P95\% & \shortstack[c]{\rule{0pt}{2.6ex}Cost\\difference} & \shortstack[c]{\rule{0pt}{2.6ex}99\% CI\\error} & \shortstack[c]{\rule{0pt}{2.6ex}Scenarios with\\load shedding} & \shortstack[c]{\rule{0pt}{2.6ex}Average\\expected return} \\
\hline
Base case & 1 & 43,683.89 & 208,239.53 & 392,959.51 & \multirow{2}{*}{179,450.41} & \multirow{2}{*}{$\pm$ 5,050.81} & 10.3\% & 98.0\% \\
 & 2 & 7,704.94 & 28,789.12 & 90,392.76 & & & 0.9\% & 18.3\% \\
\hline
Initial storage 90\% & 1 & 4,631.81 & 134,909.77 & 291,963.30 & \multirow{2}{*}{119,392.17} & \multirow{2}{*}{$\pm$ 4,462.62} & 6.5\% & 60.3\% \\
 & 2 & 5,485.86 & 15,517.61 & 53,509.11 & & & 0.5\% & 9.5\% \\
\hline
Initial storage 10\% & 1 & 64,060.82 & 232,404.68 & 421,564.20 & \multirow{2}{*}{186,091.21} & \multirow{2}{*}{$\pm$ 5,285.75} & 11.8\% & 101.9\% \\
 & 2 & 15,942.27 & 46,313.47 & 109,708.07 & & & 1.7\% & 13.8\% \\
\hline
Thermal capacity $+20\%$ & 1 & 5,141.58 & 106,078.47 & 238,804.60 & \multirow{2}{*}{91,112.80} & \multirow{2}{*}{$\pm$ 3,547.74} & 6.3\% & 67.2\% \\
 & 2 & 3,802.06 & 14,965.67 & 49,726.26 & & & 0.5\% & 12.6\% \\
\hline
Thermal capacity $-20\%$ & 1 & 123,890.21 & 354,984.40 & 602,259.98 & \multirow{2}{*}{293,837.67} & \multirow{2}{*}{$\pm$ 7,482.38} & 14.9\% & 131.8\% \\
 & 2 & 42,944.49 & 61,146.73 & 98,341.64 & & & 0.4\% & 42.8\% \\
\hline
\end{tabular}%
}
\end{table}

Interestingly, under stressed system conditions, characterized by lower initial storage levels or reduced thermal capacity, not only do the performance indicators deteriorate under Policy 1---including operating costs, expected returns and their variability, and load-shedding probability---but the performance gap between Policies 1 and 2 also widens. In some cases, the unbiased policy (Policy 2) even exhibits improvements in specific performance indicators despite the more adverse conditions, such as expected returns under lower initial storage levels and load-shedding probability under reduced thermal capacity. This is never observed under Policy 1, whose performance deteriorates more severely across the sensitivity cases. These findings provide further evidence that correcting the forecasting bias enhances system robustness in terms of reliability, operating costs, and price stability, particularly under stressed operating conditions.

\bibliographystyle{elsarticle-harv}
\bibliography{bibliography}
\end{document}